\documentclass[sigconf,edbt]{acmart-edbt2019}

\usepackage{booktabs} 

\usepackage{multirow}
\usepackage{siunitx}

\usepackage{amsmath,amsthm}
\usepackage[ruled,vlined,linesnumbered]{algorithm2e}

\usepackage{graphicx}
\usepackage{float}
\usepackage{subcaption}

\setcopyright{rightsretained}

\acmDOI{}

\acmISBN{978-3-89318-081-3}

\acmConference[EDBT 2019]{22nd International Conference on Extending Database Technology (EDBT)}{March 26-29, 2019}{Lisbon, Portugal} 
\acmYear{2019}

\begin{document}
\title{A Highly Scalable Labelling Approach for Exact Distance Queries in Complex Networks}
\titlenote{Produces the permission block, and copyright information}
  
\author{Muhammad Farhan, Qing Wang, Yu Lin, Brendan McKay} 
\affiliation{ 
    \institution{Australian National University}
    \city{Canberra}
    \country{Australia} 
}
\email{{muhammad.farhan, qing.wang, yu.lin, brendan.mckay}@anu.edu.au}

\renewcommand{\shortauthors}{}

\begin{abstract}
Answering exact shortest path distance queries is a fundamental task in graph theory. Despite a tremendous amount of research on the subject, there is still no satisfactory solution that can scale to billion-scale complex networks. Labelling-based methods are well-known for rendering fast response time to distance queries; however, existing works can only construct labelling on moderately large networks (million-scale) and cannot scale to large networks (billion-scale) due to their prohibitively large space requirements and very long preprocessing time. In this work, we present novel techniques to efficiently construct distance labelling and process exact shortest path distance queries for complex networks with billions of vertices and billions of edges. Our method is based on two ingredients: (i) a scalable labelling algorithm for constructing minimal distance labelling, and (ii) a querying framework that supports fast distance-bounded search on a sparsified graph. Thus, we first develop a novel labelling algorithm that can scale to graphs at the billion-scale. Then, we formalize a querying framework for exact distance queries, which combines our proposed highway cover distance labelling with distance-bounded searches to enable fast distance computation. To speed up the labelling construction process, we further propose a parallel labelling method that can construct labelling simultaneously for multiple landmarks. We evaluated the performance of the proposed methods on 12 real-world networks. The experiments show that the proposed methods can not only handle networks with billions of vertices, but also be up to 70 times faster in constructing labelling and save up to 90\% of labelling space. In particular, our method can answer distance queries on a billion-scale network of around 8B edges in less than 1ms, on average.
\end{abstract}

%
%


\maketitle
\section{INTRODUCTION}\label{section:intro}
Finding the shortest-path distance between a pair of vertices is a fundamental task in graph theory, and has a broad range of applications \cite{backstrom2006group,freeman1977set,sabidussi1966centrality,yahia2008efficient,ukkonen2008searching,vieira2007efficient,maniu2013network}. For example, in web graphs, ranking of web pages based on their distances to recently visited web pages helps in finding the more relevant pages and is referred to as context-aware search \cite{ukkonen2008searching}. In social network analysis, distance is used as a core measure in many problems such as centrality \cite{freeman1977set,sabidussi1966centrality} and community identification \cite{backstrom2006group}, which require distances to be computed for a large number of vertex pairs. However, despite extensive efforts in addressing the shortest-path distance problem for many years, there is still a high demand for scalable solutions that can be used to support analysis tasks over large and ever-growing networks.

Traditionally, one can use the Dijkstra algorithm \cite{tarjan1983data} for weighted graphs or a breadth-first search (BFS) algorithm for unweighted graphs to query shortest-path distances. However, these algorithms are not scalable, i.e., for large graphs with billions of vertices and edges, they may take seconds or even longer to find the shortest-path distance between one pair of vertices, which is not acceptable for large-scale network applications where distances need to be provided in the order of milliseconds. To improve query time, a well-established approach is to precompute and store shortest-path distances between all pairs of vertices in an index, also called \emph{distance labelling}, and then answer a \emph{distance query} (i.e., find the distance between two vertices) in constant time with a single lookup in the index. Recent work \cite{hayashi2016fully} shows that such labelling-based methods are the fastest known exact distance querying methods on moderately large graphs (million-scale) having millions of edges, but still fail to scale to large graphs (billion-scale) due to quadratic space requirements and unbearable indexing construction time.

\begin{figure*}[ht]
\hspace*{-0.3cm}\begin{minipage}{0.60\linewidth}
\begin{center}
(a)\hspace{5cm}(b)    
\end{center}
\vspace{-0.1cm}
\centering
\hspace*{-0.2cm}\includegraphics[scale=0.21]{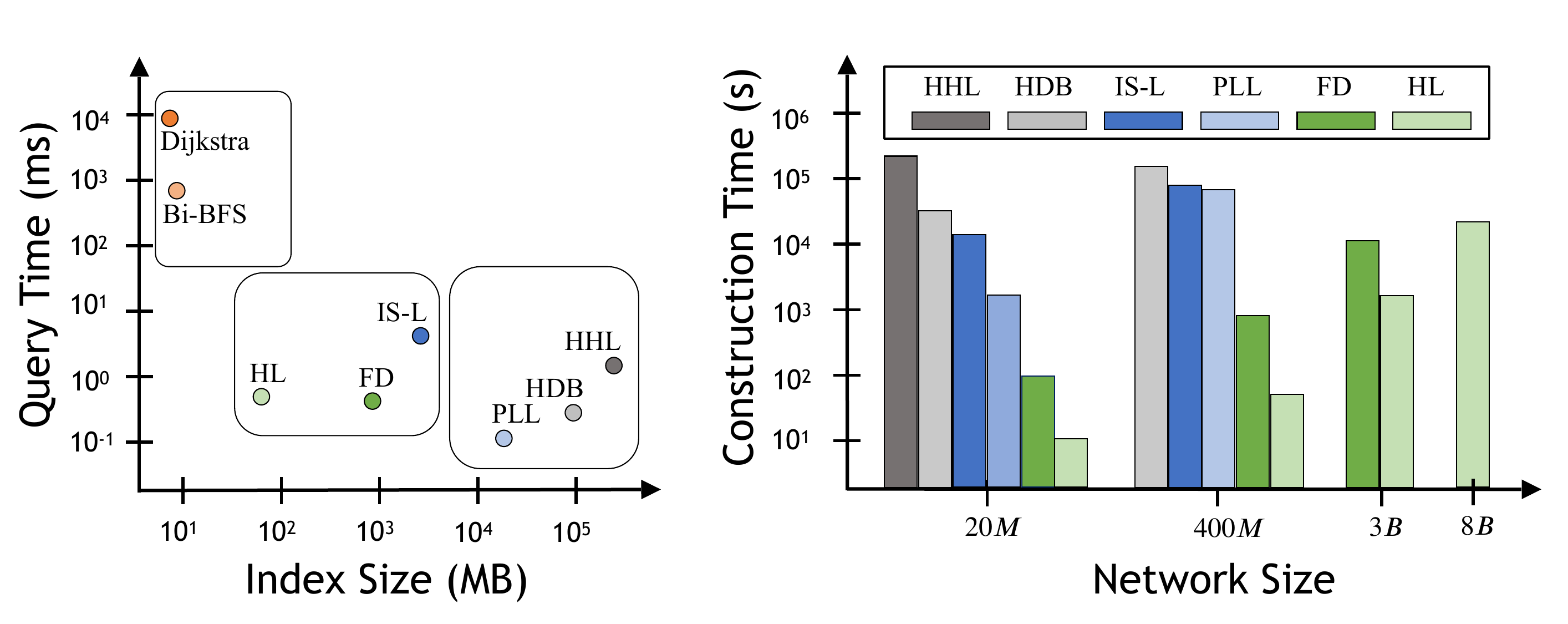}
\end{minipage}
\begin{minipage}{0.39\linewidth}
\vspace*{-0.9cm}
\begin{center}
(c) 
\end{center}
\vspace*{0.3cm}
\scalebox{0.8}{\begin{tabular}{| l || c|c| c| c |}  \hline
   \multirow{2}{*}{Method} & \textsc{Ordering}- & \textsc{2HC}-& \textsc{HWC}- & \multirow{2}{*}{\textsc{Parallel}?} \\
    &\textsc{dependent}?& \textsc{minimal}?& \textsc{minimal}?&\\
   \hline\hline
   HL (ours) & no & n/a &yes & landmarks  \\
   FD \cite{hayashi2016fully} & no & no &no & neighbours \\
   IS-L \cite{fu2013label} & yes &no & no & no \\\hline\hline
   PLL \cite{akiba2013fast} & yes & yes & no & neighbours  \\
   HDB \cite{jiang2014hop} & yes & no & no & no \\
   HHL \cite{abraham2012hierarchical} & yes & no & no & no \\
   \hline
 \end{tabular}}
\end{minipage}
\vspace{-0.2cm}
\caption{High-level overview of the state-of-the-art methods and our proposed method (HL) for exact distance queries: (a) performance w.r.t. query time and labelling size on networks of size up to 400M, (b) scalability w.r.t. labelling construction time and network size, and (c) several important properties related to labelling methods.}
\label{fig:intro}
\end{figure*}

Thus, the question is still open as to how scalable solutions to answer exact distance queries in billion-scale networks can be developed. Essentially, there are three computational factors to be considered concerning the performance of algorithms for answering distance queries: construction time, index size, and query time. Much of the existing work has focused on exploring trade-offs among these computational factors \cite{abraham2011hub,abraham2012hierarchical,akiba2013fast,akiba2012shortest,wei2010tedi,hayashi2016fully,tretyakov2011fast,potamias2009fast,fu2013label,jin2012highway,qiao2014approximate,gubichev2010fast,li2017experimental, chang2012exact}, especially for the 2-hop cover distance labelling \cite{cohen2003reachability,akiba2013fast}.  Nonetheless, to handle large graphs, we believe that a scalable solution for answering exact distance queries needs to have the following desirable characteristics: (1) the construction time of a distance labelling is scalable with the size of a network; (2) the size of a distance labelling is minimized so as to reduce the space overhead; (3) the query time remains in the order of milliseconds, even in graphs with billions of nodes and edges.

In this work, we aim to develop a scalable solution for exact distance queries which can meet the aforementioned characteristics.  Our solution is based on two ingredients: (i) a scalable labelling algorithm for constructing minimal distance labelling, and (ii) a querying framework that supports fast distance-bounded search on a sparsified graph. More specifically, we first develop a novel labelling algorithm that can scale to graphs at the billion-scale. We observed that, for a given number of landmarks, the distance entries from these landmarks to other vertices in a graph can be further minimized if the definition of 2-hop cover distance labelling is relaxed. 
Thus, we formulate a relaxed notion for labelling in this paper, called the \emph{highway cover distance labelling}, and develop a simple yet scalable labelling algorithm that adds a significantly small number of distance entries into the label of each vertex. We prove that the distance labelling constructed by our labelling algorithm is minimal, and also experimentally verify that the construction process is scalable.

Then, we formalize a querying framework for exact distance queries, which combines our proposed highway cover distance labelling with distance-bounded searches to enable fast distance computation. This querying framework is capable of balancing the trade-off between construction time, index size and query time through an offline component (i.e. the proposed highway cover distance labelling) and an online component (i.e. distance-bounded searches). The basic idea is to select a small number of highly central landmarks that allow us to efficiently compute the upper bounds of distances between all pairs of vertices using an offline distance labelling, and then conduct distance-bounded search over a sparsified graph to find exact distances efficiently.  Our experimental results show that the query time of distance queries within this framework is still in millionseconds for large graphs with billions of vertices and edges.

Figure \ref{fig:intro} summarizes the performance of the state-of-the-art methods for exact distance queries \cite{akiba2013fast, fu2013label, abraham2012hierarchical, hayashi2016fully, jiang2014hop, tarjan1983data, pohl1971bi, chang2012exact}, as well as our proposed method in this paper, denoted as HL. In Figure \ref{fig:intro}(a)-\ref{fig:intro}(b), 
we can see that, labelling-based methods PLL \cite{akiba2013fast}, HDB \cite{jiang2014hop}, and HHL \cite{abraham2012hierarchical} can answer distance queries in a considerably small amount of time. However, they have very large space requirements and very long labelling construction time. On the contrary, traditional online search methods such as Dijkstra \cite{tarjan1983data} and bidirectional BFS (denoted as Bi-BFS) \cite{pohl1971bi} are not applicable to large-scale networks where distances need to be provided in the order of milliseconds because of their very high response time. The hybrid methods FD \cite{fu2013label}, IS-L \cite{hayashi2016fully} and HL (our method) combine an offline labelling and an online graph traversal technique, which can provide better trade-offs between query response time and labelling size. In Figure \ref{fig:intro}(b), we can also see that only our proposed method HL can handle networks of size 8B, and is scalable to perform distance queries on networks with billions of vertices and billions of edges. 

Figure \ref{fig:intro}(c) presents a high-level overview for several important properties of labelling methods. The column \textsc{ordering dependent} refers to whether a distance labelling depends on the ordering of landmarks when being constructed by a method. Only our method HL and FD are not ordering-dependent. The columns \textsc{2HC-minimal} and \textsc{HWC-minimal} refer to whether a distance labelling constructed by a method is minimal in terms of the 2-hop cover (2HC) and highway cover (HWC) properties, respectively. PLL is 2HC-minimal, but not HWC-minimal. Our method HL is the only method that is HWC-minimal. The column \textsc{Parallel} refers to what kind of parallelism a method can support. FD and PLL support bit-parallelism for up to 64 neighbours of a landmark. Our method HL supports parallel computation for multiple landmarks, depending on the number of processors. Other methods did not mention any parallelism.

In summary, our contributions in this paper are as follows:
\begin{itemize}
\item We introduce a new labelling property, namely highway cover labelling, which relaxes the standard notion of 2-hop cover labelling. Based on this new labelling property, we propose a highly scalable labelling algorithm that can scale to construct labellings for graphs with billions of vertices and billions of edges. 

\item We prove that the proposed labelling algorithm can construct HWC-minimal labellings, which is independent of any ordering of landmarks. Then, due to this determinstric nature of labelling, we further develop a parallel algorithm which can run parallel BFSs from multiple landmarks to speed up labelling construction.

\item We combine our novel labelling algorithm with online bounded-distance graph traversal to efficiently answer exact distance queries. This querying framework enables us to balance the trade-offs among construction time, labelling size and query time.

\item We have experimentally verified the performance of our methods on 12 large-scale complex networks. The results show that our methods can not only handle networks with billions of vertices, but also be up to 70 times faster in constructing labelling and save up to 90\% of labelling space. 
\end{itemize}


The rest of the paper is organized as follows. In Section \ref{section:preliminaries}, we present basic notations and definitions used in this paper. Then, we discuss a novel labelling algorithm in Section \ref{section:labelling}, formulate the querying framework in Section \ref{section:querying}, and introduce several optimization techniques in Section \ref{section:optimization}. In Section \ref{section:experiments} we present our experimental results and in Section \ref{section:background} we discuss other works that are related to our work here. The paper is concluded in Section \ref{section:conclusion}.

\section{Preliminaries}\label{section:preliminaries}
Let $G = (V, E)$ be a graph where $V$ is a set of vertices and $E \subseteq V \times V$ is a set of edges. We have $n = |V|$ and $m = |E|$. Without loss of generality, we assume that the graph $G$ is connected and undirected in this paper. Let $V' \subseteq V$ be a subset of vertices of $G$. Then the induced subgraph $G[V']$ is a graph whose vertex set is $V'$ and whose edge set consists of all of the edges in $E$ that have both endpoints in $V'$. Let $N_G(v) = \{u \in V | (u, v) \in E \}$ denote a set of neighbors of a vertex $v \in V$ in $G$. 

The \emph{distance} between two vertices $s$ and $t$ in $G$, denoted as $d_G(s, t)$, is the length of the shortest path from $s$ to $t$. We consider $d_G(s, t) = \infty$, if there does not exist a path from $s$ to $t$. For any three vertices $s, u, t \in V$, the following \emph{triangle inequalities} are satisfied: 
\begin{align} \label{distance_metric}
	d_G(s, t) \leq d_G(s, u) + d_G(u, t) \\
    d_G(s, t) \geq |d_G(s, u) - d_G(u, t)|
\end{align}
If $u$ belongs to one of the shortest paths from $s$ to $t$, then $d_G(s, t) = d_G(s, u) + d_G(u, t)$ holds.

Given a special subset of vertices $R\subseteq V$ of $G$, so-called \emph{landmarks}, a \emph{label} $L(v)$ for each vertex $v \in V$ can be precomputed, which is a set of \emph{distance entries} $\{(u_1, \delta_L(u_1, v)), \dots, (u_n, \delta_L(u_n,v)\\)\}$ where  $u_i \in R$ and $\delta_L(u_i, v) = d_G(u_i, v)$ for $i=1,\dots, n$. The set of labels $L = \{L(v)\}_{v \in V}$ is called a \emph{distance labeling} over $G$. The \emph{size} of a distance labelling $L$ is defined as size(L)=$\Sigma_{v\in V}|L(v)|$. 

Using such a \emph{distance labeling L}, we can query the distance between any pair of vertices $s, t \in V$ in graph $G$ as follows,

\label{eq:twoHop_query_distance}
\begin{align}
\text{Q}(s, t, L) = \texttt{min}\{\delta_L(u, s) + \delta_L(u, t) | (u, \delta_L(u, s)) \in  L(s),\notag\\ (u, \delta_L(u, t)) \in  L(t)\}
\end{align}
We define $\text{Q}(s, t, L) = \infty$, if $L(s)$ and $L(t)$ do not share any landmark. If $\text{Q}(s, t, L) = d_G(s, t)$ holds for any two vertices $s$ and $t$ of $G$, $L$ is called a \emph{2-hop cover distance labeling} over $G$ \cite{cohen2003reachability,abraham2012hierarchical}.

Given a graph $G$ and a set of landmarks $R \subseteq V$, the \emph{distance querying problem} is to efficiently compute the shortest path distance $d_G(s,t)$ between any two vertices $s$ and $t$ in $G$, using a distance labeling $L$ over $G$ in which labels may contain distance entries from landmarks in $R$.
\section{Highway Cover Labelling}\label{section:labelling}
In this section, we formulate the highway cover labelling problem and propose a novel algorithm to efficiently construct the highway cover distance labelling over graphs. Then, we provide theoretical analysis of our proposed algorithm.

\subsection{Highway Cover Labelling Problem}

We begin with the definitions of highway and highway cover. 

\begin{definition}(\emph{Highway})~ A highway $H$ is a pair $(R, \delta_H)$, where $R$ is a set of landmarks and $\delta_H$ is a $distance$ $decoding$ $function$, i.e. $\delta_H : R \times R \rightarrow \mathbb{N}^+$, such that for any $\{r_1, r_2\} \subseteq R$ we have $\delta_H(r_1, r_2) = d_G(r_1, r_2)$.
\end{definition}

Given a landmark $r\in R$ and two vertices $s, t\in V\backslash R$ (i.e. $V\backslash R=V-R$), a \emph{$r$-constrained shortest path} between $s$ and $t$ is a path between $s$ and $t$ satisfying two conditions: (1) It goes through the landmark $r$, and (2) It has the minimum length among all paths between $s$ and $t$ that go through $r$. We use $P_{st}$ to denote the set of vertices in a shortest path between $s$ and $t$, and $P^r_{st}$  to denote the set of vertices in a $r$-constrained shortest path between $s$ and $t$.

\begin{definition}\label{def:highway-cover}(\emph{Highway Cover})~ Let $G=(V,E)$ be a graph and $H=(R, \delta_H)$ a highway. Then for any two vertices $s, t \in V\backslash R$ and for any $r\in R$, there exist $(r_i, \delta_L(r_i, s)) \in L(s)$ and $(r_j, \delta_L(r_j, t)) \in L(t)$ such that $r_i \in P_{rs}$ and $r_j \in P_{rt}$, where $r_i$ and $r_j$ may equal to $r$.

\end{definition}

If the label of a vertex $v$ contains a distance entry $(r, \delta_L(r, v))$, we also say that the vertex $v$ is \emph{covered} by the landmark $r$ in the distance labelling $L$. Intuitively, the highway cover property guarantees that, given a highway $H$ with a set of landmarks $R$ and $r\in R$, any $r$-constrained shortest path distance between two vertices $s$ and $t$ can be found using only the labels of these two vertices and the given highway. A distance labelling $L$ is called a \emph{highway cover distance labelling} if $L$ satisfies the highway cover property. \looseness=-1

\begin{example}
Consider the graph $G$ depicted in Figure \ref{fig:highway-cover-labelling}(a), the highway $H$ has three landmarks $\{1,5,9\}$ as highlighted in red in Figure \ref{fig:highway-cover-labelling}(b). Based on the graph in Figure \ref{fig:highway-cover-labelling}(a) and the highway in Figure \ref{fig:highway-cover-labelling}(b), we have $\langle 11, 1, 4 \rangle$ which is a shortest path between the vertices $11$ and $4$ constrained by the landmark $1$, i.e. $1$-constrained shortest path between $11$ and $4$. In contrast, neither of the paths $\langle 11, 10, 9, 1, 4\rangle$ and $\langle 11, 4\rangle$ is a $1$-constrained shortest path between $11$ and $4$.

In Figure \ref{fig:highway-cover-labelling}(b), the outgoing arrows from each landmark point to vertices in $G$ that are covered by this landmark in the highway. The distance labelling in Figure \ref{fig:highway-cover-labelling}(c) satisfies the highway cover property because for any two vertices that are not landmarks and any landmark $r\in\{1,5,9\}$, we can find the $r$-constrained shortest path distance between these two vertices using their labels and the highway.



\begin{figure}[h]
\centering
\includegraphics[scale=0.3]{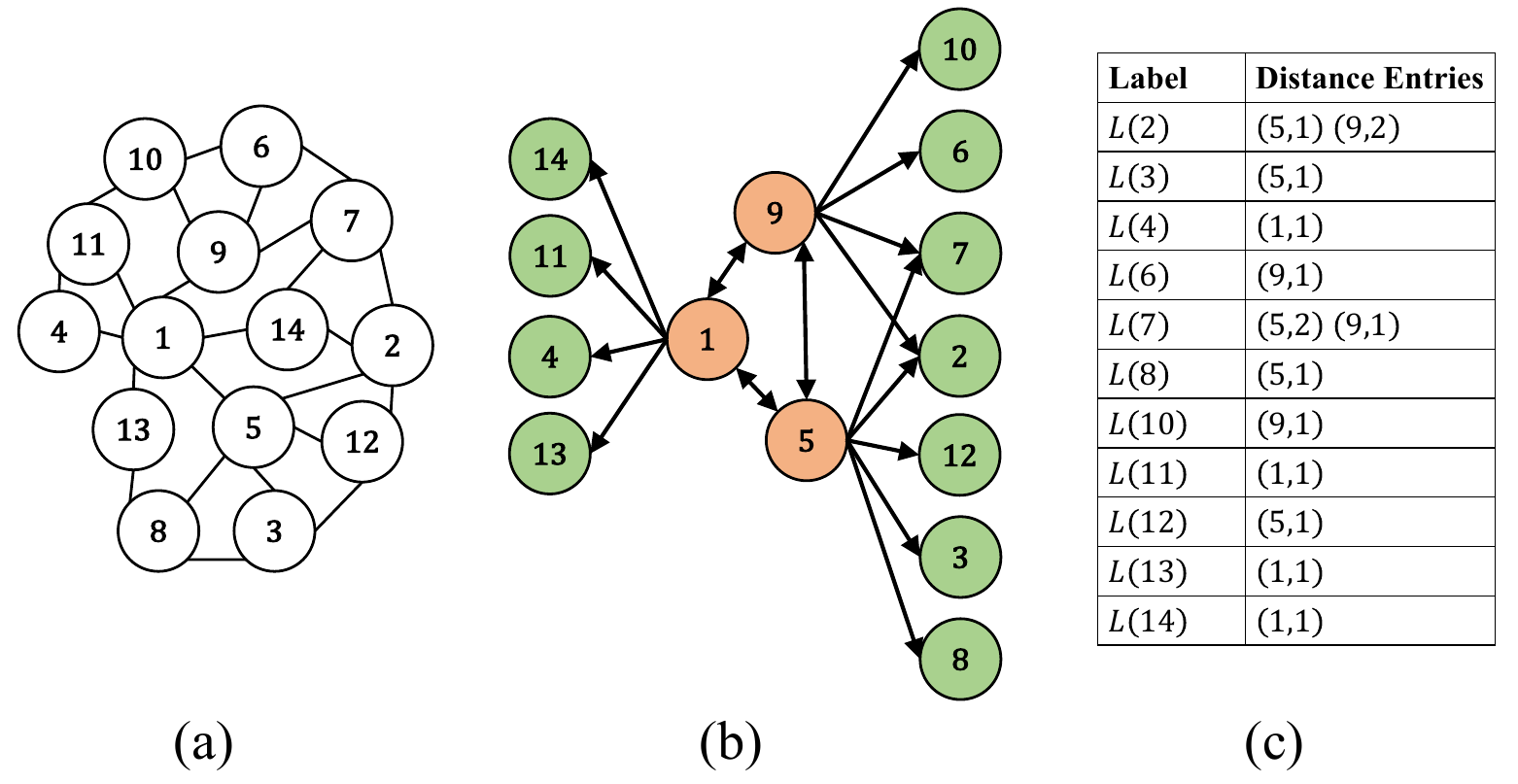}
\vspace{-0.2cm}
\caption{An illustration of highway cover distance labelling: (a) an example graph $G$, (b) a highway structure $H$ and (c) a distance labelling that fulfills the highway cover property over $(G,H)$.}
\label{fig:highway-cover-labelling}
\end{figure}
\end{example}


\begin{definition}(\emph{Highway Cover Labelling Problem})
Given a graph $G$ and a highway $H$ over $G$, the \emph{highway cover labelling problem} is to efficiently construct a highway cover distance labelling $L$. 
\end{definition}

Several choices naturally come up: (1) One is to add a distance entry for each landmark into the label of every vertex in $V-R$, as the approach proposed in \cite{hayashi2016fully}; (2) Another is to use the pruned landmark labelling approach \cite{akiba2013fast} to add the distance entry of a landmark $r$ into the labels of vertices in $V-R$ if the landmark has not been pruned during a BFS rooted at $r$; (3) Alternatively, we can also extend the pruned landmark labelling approach to construct the highway cover labeling by replacing the 2-hop cover pruning condition with the one required by the highway cover as defined in Definition \ref{def:highway-cover} at each step of checking possible labels to be pruned. 

In all these cases, the labelling construction process would not be scalable nor be suitable for large-scale complex networks with billions of vertices and edges. Moreover, these approaches would potentially lead to the construction of distance labellings with different sizes. A question arising naturally is how to construct a minimal highway cover distance labelling without redundant labels?
In a nutshell, it is a challenging task to construct a highway cover distance labelling that can scale to very large networks, ideally in linear time, but also with the minimal labelling size.


\subsection{A Novel Algorithm}
We propose a novel algorithm for solving the highway cover labelling problem, which can construct labellings in linear time. 

The key idea of our algorithm is to construct a \emph{label} $L(v)$ for vertex $v \in V \backslash R$ such that the distance entry $(r_i, \delta_L(r_i, v))$ of each landmark $r_i\in R$ is only added into the label $L(v)$ iff there does not exist any other landmark that appears in the shortest path between $r_i$ and $v$, i.e. $P_{r_iv}\cap R=\{r_i\}$. In other words, if there is another landmark $r\in R$ and $r_i$ is in the shortest path between $r$ and $v$, then $(r_i, \delta_L(r_i, v))$ is added into $L(v)$ iff $r_i$ is the ``closest" landmark to $v$. 
To compute such labels efficiently, we conduct a breadth-first search from every landmark $r_i \in R$ and add distance entries into labels of vertices that do not have any other landmark in their shortest paths from $r_i$.

\begin{example}
Consider vertex $7$ in Figure 2(c), the label $L(7)$ contains the distance entries of landmarks $\{5, 9\}$, but no distance entry of landmark $1$. This is because $5$ and $9$ are the closest landmarks to vertex 7 in the shortest paths $\langle 5, 7\rangle$ and $\langle 9, 7\rangle$, respectively. However, for either of two shortest paths $\langle 1, 9, 7\rangle$ and $\langle 1,5,7\rangle$ between $1$ and $7$, there is another landmark (i.e. $5$ or $9$) that is closer to $7$ compared with $1$ in these shortest paths. Thus the distance entry of landmark 1 is not added into $L(7)$.
\end{example}
Our highway cover labelling approach is described in Algorithm \ref{algo:wholePreprocessingAlgo}. 
Given a graph $G$ and a highway $H$ over $G$, we start with an empty \emph{highway cover distance labelling} $L$, where $L(v) = \emptyset$ for every $v \in V \backslash R$. Then, for each landmark $r_i \in R$, we compute the corresponding distance entries as follows. We use two queues $\mathcal{Q}_{label}$ and $\mathcal{Q}_{prune}$ to process vertices to be labeled or pruned at each level of a breadth-first search (BFS) tree, respectively. We start by processing vertices in $\mathcal{Q}_{label}$. For each vertex $u \in \mathcal{Q}_{label}$ at depth $n$, we examine the children of $u$ at depth $n + 1$ that are unvisited. For each unvisited child vertex $v \in N_G(u)$ at depth $n+ 1$, if $v \in R$ then we prune $v$, i.e., we do not add a distance entry of the current landmark $r_i$ into $L(v)$ and we also enqueue $v$ to the pruned queue $\mathcal{Q}_{prune}$ (Line 11). Otherwise, we add $(r_i, \delta_{BFS}(r_i, v))$ to the \emph{label} of $v$, i.e., we add it into $L(v)$ and we also enqueue $v$ to the labeled queue $\mathcal{Q}_{label}$ (Lines 13-14). Here, $\delta_{BFS}(r_i, v)$ refers to BFS decoded distance from root $r_i$ to $v$. Then we process the pruned vertices in $\mathcal{Q}_{prune}$. These vertices are either landmarks or have landmarks in their shortest paths from $r_i$, and thus do not need to be labeled. Therefore, for each vertex $v \in \mathcal{Q}_{prune}$ at depth $n$, we enqueue all unvisited children of $v$ at depth $n + 1$ to the pruned queue $\mathcal{Q}_{prune}$. We keep processing these two queues, one after the other, until $\mathcal{Q}_{label}$ is empty. \looseness=-1 

\begin{algorithm}
 \SetAlgoLined
 \caption{Constructing the highway cover labelling $L$}
 \label{algo:wholePreprocessingAlgo}
 \KwIn{$G = (V. E)$, $H = (R, \delta_H)$}
 \KwOut{$L$}
 $L(v) \gets \emptyset, \forall v \in V \backslash R$ \\
 \ForEach{$r_i \in R$}{
 	$\mathcal{Q}_{label} \gets \emptyset$ \\
    $\mathcal{Q}_{prune} \gets \emptyset$ \\
    $n \gets 0$ \\
    Enqueue $r_i$ to $\mathcal{Q}_{label}$ and set $r_i$ as the root of BFS \\
    \While{$\mathcal{Q}_{label}$ is not empty}{
      \ForEach{$u \in \mathcal{Q}_{label}$ at depth $n$}{
        \ForEach{unvisited child $v$ of $u$ at depth $n + 1$}{
          \uIf{$v$ is a landmark}{
            Enqueue $v$ to $\mathcal{Q}_{prune}$
          }
          \Else{
            Enqueue $v$ to $\mathcal{Q}_{label}$ \\
            Add \{($r_i$, $\delta_{BFS}(r_i, v)$)\} to $L(v)$
          }
        }
      }
      $n \gets n + 1$ \\
      \ForEach{$v \in \mathcal{Q}_{prune}$ at depth $n$}{
        Enqueue unvisited children of $v$ at depth $n + 1$ to $\mathcal{Q}_{prune}$
      }
    }
 }
 \textbf{return} $L$
\end{algorithm}


\begin{example}
We illustrate how our algorithm conducts pruned BFSs in Figure \ref{fig:highway-cover-algoritm}. The pruned BFS from landmark $1$ is depicted in Figure \ref{fig:highway-cover-algoritm}(a), which labels only four vertices $\{4, 11, 13, 14\}$ because the other vertices are either landmarks or contain other landmarks in their shortest paths to landmark $1$. Similarly, in the pruned BFS from landmark $5$ depicted in Figure \ref{fig:highway-cover-algoritm}(b), only vertices $\{7,2,12,3,8\}$ are labelled, and none of the vertices $4$, $11$, $13$ and $14$ is labelled because of the presence of landmark $1$ in their shortest paths to landmark $5$. Indeed, we can get the distance between landmark $5$ to these vertices by using the highway, i.e. $\delta_{H}(5,1)$, and distance entries in their labels to landmark $1$. The pruned BFS from landmark 9 is depicted in Figure \ref{fig:highway-cover-algoritm}(c), which works in a similar fashion.
\end{example}

Note that, although a highway $H$ is given in Algorithm \ref{algo:wholePreprocessingAlgo}, we can indeed compute the distances $\delta_H$ for a given set of landmarks $R$ along with Algorithm \ref{algo:wholePreprocessingAlgo}.

\begin{figure}[h]
\centering
\includegraphics[scale=0.3]{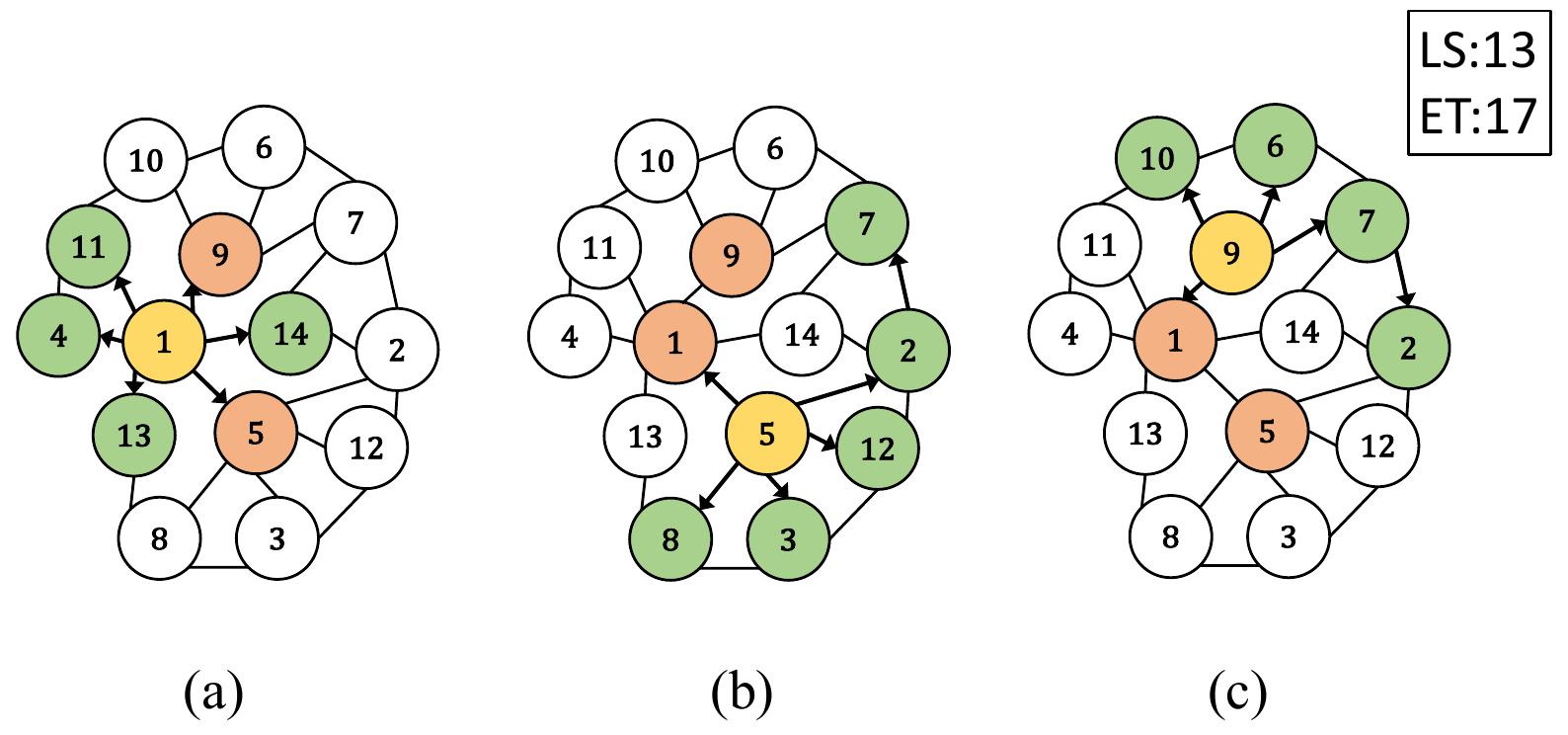}
\caption{An illustration of the highway cover labelling algorithm: (a), (b) and (c) describe the pruned BFSs that are rooted at the landmarks 1, 5 and 9, respectively, where yellow vertices denote roots, green vertices denote those being labeled, red vertices denote landmarks, and white vertices are not labelled. LS and ET at the top right corner denote the labelling size and the number of edges traversed during the pruned BFSs, respectively.}
\label{fig:highway-cover-algoritm}
\end{figure}

\subsection{Correctness} 
Here we prove the correctness of our labelling algorithm. 

\begin{lemma}\label{lemma:p1}
In Algorithm \ref{algo:wholePreprocessingAlgo}, for each pruned BFS rooted at $r_i \in R$,  $(r_i, \delta_L(r_i, v))$ is added into the label of a vertex $v\in V\backslash R$ iff there is no any other landmark appearing in the shortest path between $r_i$ and $v$, i.e., $P_{r_iv}\cap R=\{r_i\}$.
\end{lemma}

\begin{proof}

Suppose that Algorithm \ref{algo:wholePreprocessingAlgo} is conducting a pruned BFS rooted at $r_i$ and $v$ is an unvisited child of another vertex $u$ in $\mathcal{Q}_{label}$ (start from $\mathcal{Q}_{label}=\{r_i\}$) (Lines 6-9). If $v\in R$ (Line 10), then we have $(P_{r_iw}\cap R)\supseteq\{r_i,v\}$ (Lines 11, 19-21), $(r_i, \delta_L(r_i, w))$ cannot be added into the label of any child $w$ of $v$, i.e., put $w$ into $\mathcal{Q}_{prune}$. Otherwise, by $v\notin R$ and $v$ is an unvisited child of a vertex $u$ in $\mathcal{Q}_{label}$ (Lines 8-9), we know that $P_{r_iv}\cap R=\{r_i\}$ and thus $(r_i, \delta_L(r_i, v))$ is added into $L(v)$ (lines 12-14). 
\end{proof}

Then, by Lemma \ref{lemma:p1}, we have the following corollary.
\begin{corollary}\label{lemma:p2}
Let $r \in R$ be a landmark, $v \in V\backslash R$ a vertex, and $L$ a distance labelling constructed by Algorithm \ref{algo:wholePreprocessingAlgo}, if $(r, \delta_L(r, v)) \notin L(v)$, then there must exist a landmark $r_j$ such that $(r_j, \delta_L(r_j, v)) \in L(v)$ and $d_G(r, v) = \delta_L(r_j, v) + \delta_H(r, r_j)$.
\end{corollary}

\begin{theorem}
The highway cover distance labelling $L$ over $(G,H)$ constructed using Algorithm \ref{algo:wholePreprocessingAlgo} satisfies the highway cover property over $(G,H)$.
\end{theorem}

\begin{proof}
To prove that, for any two vertices $s, t \in V\backslash R$ and for any $r\in R$, there exist $(r_i, \delta_L(r_i, s)) \in L(s)$ and $(r_j, \delta_L(r_j, t)) \in L(t)$ such that $r_i \in P_{rs}$ and $r_j \in P_{rt}$, we consider the following 4 cases: (1) If $r \in L(s)$ and $r \in L(t)$, then $r = r_i= r_j$. 
(2) If $r \in L(s)$ and $r \notin L(t)$, then $r_i=r$ and by Lemma \ref{lemma:p2}, there exists another landmark $r_j$ such that $r_j$ is in the shortest path between $t$ and $r$ and $(r_j, \delta_L(r_j, t)) \in L(t)$. (3) If $r \notin L(s)$ and $r \in L(t)$, then similarly we have $r_j=r$, and by Lemma \ref{lemma:p2}, there exists another landmark $r_i$ such that $r_i$ is in the shortest path between $s$ and $r$ and $(r_i, \delta_L(r_i, s)) \in L(s)$.
(4) If $r \notin L(s)$ and $r \notin L(t)$, then by Lemma \ref{lemma:p2} there exist another two landmarks $r_i$ and $r_j$ such that $r_i$ is in the shortest path between $s$ and $r$ and $(r_i, \delta_L(r_i, s)) \in L(s)$, and $r_j$ is in the shortest path between $t$ and $r$ and $(r_j, \delta_L(r_j, t)) \in L(t)$. The proof is done.
\end{proof}


\begin{figure*}[ht]
\centering
\begin{minipage}[b]{0.50\linewidth}
\centering
\includegraphics[width=\textwidth]{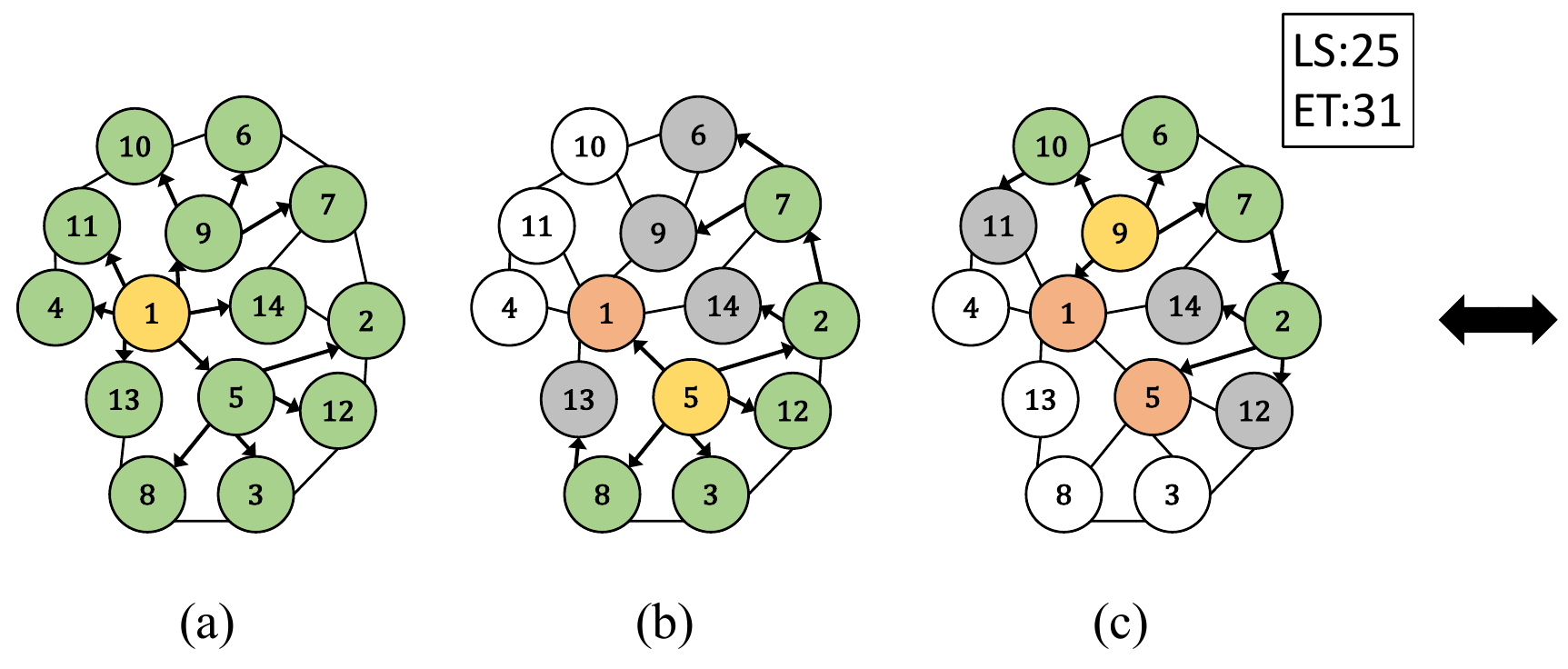}
\end{minipage}
\hspace{0.2cm}
\begin{minipage}[b]{0.46\linewidth}
\centering
\includegraphics[width=\textwidth]{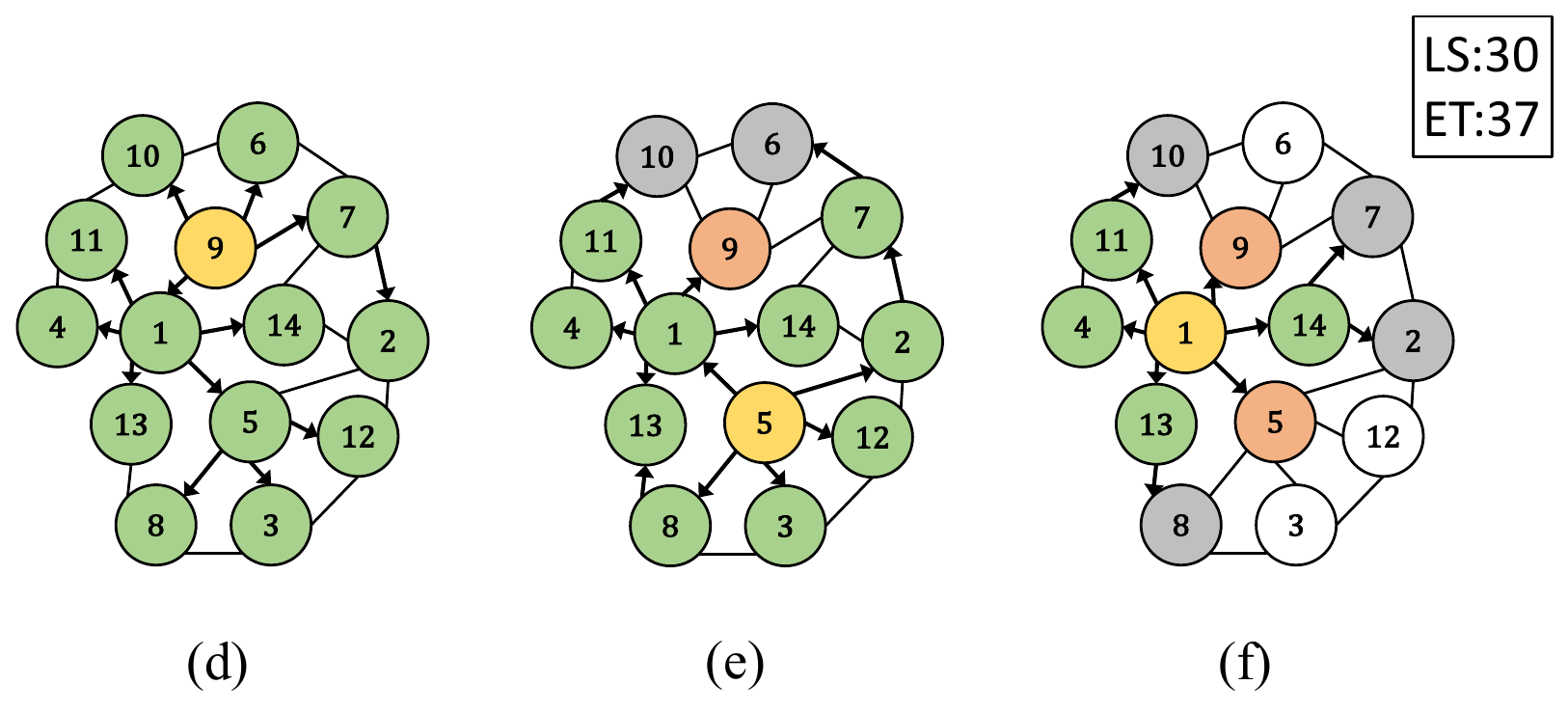}
\end{minipage}
\vspace{-0.3cm}
\caption{An illustration of the pruned landmark labelling algorithm \cite{akiba2013fast}: (a)-(c) show an example of constructing labels through pruned BFSs from three landmarks in the labelling order $\langle 1,5,9 \rangle$; (d)-(f) show an example of constructing labels using the same three landmarks but in a different labelling order $\langle 9, 5, 1\rangle$. Yellow vertices denote landmarks that are the roots of pruned BFSs, green vertices denote those being labeled, grey vertices denote vertices being visited but pruned, and red vertices denote landmarks which have already been visited.}
\label{fig:pruned-landmark-labelling-algoritm}
\end{figure*}

\subsection{Order Independence}
In previous studies \cite{abraham2011hub,akiba2013fast,abraham2012hierarchical,cohen2003reachability}, given a graph $G$, a distance labelling algorithm builds a unique canonical distance labelling subject to a labelling order (i.e., the order of landmarks used for constructing a distance labelling). It has been well known that such a labelling order is decisive in determining the size of the constructed distance labelling \cite{qin2017efficient}. For the same set of landmarks, when using different labelling orders, the sizes of the constructed distance labelling may vary significantly. 

The following example shows how different labelling orders in the pruned landmark labelling approach \cite{akiba2013fast} can lead to the distance labelling of different sizes.

\begin{example}
In Figure \ref{fig:pruned-landmark-labelling-algoritm}, the size of the distance labelling constructed using the labelling order $\langle 1,5,9 \rangle$ in Figure \ref{fig:pruned-landmark-labelling-algoritm}(a)-\ref{fig:pruned-landmark-labelling-algoritm}(c) is different from the size of the distance labelling constructed using the labelling order $\langle 9,5,1 \rangle$ in Figure \ref{fig:pruned-landmark-labelling-algoritm}(d)-\ref{fig:pruned-landmark-labelling-algoritm}(f). In both cases, the first BFS adds the distance entry of the current landmark into all the vertices in the graph. Then, the following BFSs check each visited vertex whether the shortest path distance between the current landmark and the visited vertex can be computed via the 2-hop cover property based on their labels added by the previous BFSs. A distance entry is only added into the label of a vertex if the shortest path distance cannot be computed by applying the 2-hop cover over the existing labels. Thus, the choice of the labelling order could affect the size of labels significantly. Take the vertex $11$ for example, its label contains only one distance entry $(1, 1)$ using the labelling order depicted in Figure \ref{fig:pruned-landmark-labelling-algoritm}(a)-\ref{fig:pruned-landmark-labelling-algoritm}(c), but contains three distance entries $(1.1)$, $(5,2)$, and $(9,2)$ when the labelling order depicted in Figure \ref{fig:pruned-landmark-labelling-algoritm}(d)-\ref{fig:pruned-landmark-labelling-algoritm}(f) is used.
\end{example}

Unlike all previous approaches taken with distance labelling, our highway cover labelling algorithm is order-invariant. That is, regardless of the labelling order, the distance labellings constructed by our algorithm using different labelling orders over the same set of landmarks always have the same size. In fact, we can show that our algorithm has the following stronger property: the distance labelling constructed using our algorithm is deterministic (i.e., the same label for each vertex) for a given set of landmarks.

\begin{lemma}
Let $G=(V,E)$ be a graph and $H=(R,\delta_H)$ a highway over $G$. For any two different labelling orders over $R$, the highway cover distance labellings $L_1$ and $L_2$ over $(G,H)$ constructed by these two different labelling orders using Algorithm \ref{algo:wholePreprocessingAlgo} satisfy $L_1(v)=L_2(v)$ for every $v\in V \backslash R$.
\end{lemma}

\begin{proof}
Let $O_{L_1}$ and $O_{L_2}$ be two different labelling orders over $R$. For any landmark $r$ in $O_{L_1}$ and $O_{L_2}$, Algorithm \ref{algo:wholePreprocessingAlgo} generates exactly the same pruned BFS tree. This implies that, for each vertex $v\in V \backslash R$, either the same distance entry $(r, \delta_{BFS}(r, v)$) is added into $L_1(v)$ and $L_2(v)$, or no distance entry is added to $L_1(v)$ and $L_2(v)$. Thus, Algorithm \ref{algo:wholePreprocessingAlgo} satisfy $L_1(v)=L_2(v)$ for every $v\in V \backslash R$. 
\end{proof}

\subsection{Minimality}
Here we discuss the question of minimality, i.e., whether the highway cover distance labelling constructed by our algorithm is always minimal in terms of the labelling size. We first prove the following theorem.

\begin{theorem}\label{the:minimal}
The highway cover distance labelling $L$ over $(G,H)$ constructed using Algorithm \ref{algo:wholePreprocessingAlgo} is minimal, i.e., for any highway cover distance labelling $L'$ over $(G,H)$, $size(L')\geq size(L)$ must hold. 
\end{theorem}

\begin{proof}
We prove this by contradiction. Let us assume that there is a highway cover distance labelling $L'$ with $size(L') < size(L)$. Then, this would imply that there must exist a vertex $v \in V\backslash R$ and a landmark $r \in R$ such that $r \in L(v)$ and $r \notin L'(v)$. By Lemma \ref{lemma:p1} and $r \in L(v)$, we know that there is no any other landmark in $R$ that is in the shortest path between $r$ and $v$. However, by the definition of the highway cover property (i.e. Definition \ref{def:highway-cover}) and $r \notin L'(v)$, we also know that there must exist another landmark $(r_i, \delta_L(r_i, v)) \in L(v)$ and $r_i \in P_{rv}$, which contradicts with the previous conclusion that there is no any other landmark in the shortest path between $r$ and $v$. Thus, $size(L') \geq size(L)$ must hold for any highway cover distance labelling $L'$.
\end{proof}

The state-of-the-art approaches for distance labelling is primarily based on the idea of 2-hop cover \cite{akiba2013fast, fu2013label, abraham2011hub}. One may ask the question: how is the highway cover labelling different from the 2-hop cover labelling, such as the pruned landmark labelling \cite{akiba2013fast}? It is easy to verify the following lemma that each pruned landmark labelling satisfies the highway cover property for the same set of landmarks. 

\begin{lemma}
Let $L$ be a pruned landmark labelling over graph $G$ constructed using a set of landmarks $R$. Then $L$ also satisfies the highway cover property over $(G,H)$ where $H=(R, \delta_H)$.
\end{lemma}

As the pruned landmark labelling algorithm \cite{akiba2013fast} prunes labels based on the 2-hop cover property, but our highway cover labeling algorithm prunes labels based on the property described in Lemma \ref{lemma:p1}, by Theorem \ref{the:minimal}, we also have the following corollary, stating that, for the same set of landmarks, the size of the highway cover labelling is always smaller than the size of any pruned landmark labelling.

\begin{corollary}
For a highway cover distance labelling $L_1$ produced by Algorithm \ref{algo:wholePreprocessingAlgo} over $(G,H)$, where $H=(R,\delta_H)$, and a pruned landmark labelling $L_2$ over $G$ using any labelling order over $R$, we always have $|L_1| \leq |L_2|$.
\end{corollary}

\begin{example}
Figure \ref{fig:pruned-landmark-labelling-algoritm} shows the labelling size (LS) of the pruned landmark labelling at the top right corner, which is constructed using two different orderings. The first ordering $\langle 1,5,9 \rangle$ labels 25 vertices whereas the second ordering $\langle 9,5,1 \rangle$ labels 30 vertices. On the other hand, the LS of the highway cover distance labelling is 13 as shown in Figure \ref{fig:highway-cover-algoritm}. Note that the LS of the highway cover distance labelling does not change, irrespective of ordering. Since the highway cover distance labelling constructed by our algorithm is always minimal, the LS of the highway cover distance labelling in Figure \ref{fig:highway-cover-algoritm} is much smaller than the LS of either pruned landmark labelling in Figure \ref{fig:pruned-landmark-labelling-algoritm}.
\end{example}

\section{Bounded Distance Querying}\label{section:querying}
In this section, we describe a bounded distance querying framework that allows us to efficiently compute exact shortest-path distances between two arbitrary vertices in a massive network. 

\subsection{Querying Framework}

We start with presenting a high-level overview of our querying framework. To compute the shortest path distance between two vertices $s$ and $t$ in graph $G$, our querying framework proceeds in two steps: (1) an upper bound of the shortest path distance between $s$ to $t$ is computed using the highway cover distance labelling; (2) the exact shortest path distance between $s$ to $t$ is computed using a distance-bounded shortest-path search over a sparsified graph from $G$. 


Given a graph $G$ and a highway $H=(R, \delta_H)$ over $G$, we can precompute a highway cover distance labelling $L$ using the landmarks in $R$, which enables us to efficiently compute the length of any $r$-constrained shortest path between two vertices in $V \backslash R$. The length of such a $r$-constrained shortest path must be greater than or equal to the exact shortest path distance between these two vertices and can thus serve as an upper bound in Step (1). On the other hand, since the length of such a $r$-constrained shortest path between two vertices in $V \backslash R$ can always be efficiently computed by the highway cover distance labelling $L$, the distance-bounded shortest-path search only needs to be conducted over a sparsified graph $G'$ by removing all landmarks in $R$ from $G$, i.e. $G'=G[V \backslash R]$.

More precisely, we define the bounded distance querying problem in the following.

\begin{definition}(\emph{Bounded Distance Querying Problem})
Given a sparsified graph $G'=(V',E')$, a pair of vertices $\{s,t\} \in V'$, and an upper (distance) bound $d^{\top}_{st}$, the \emph{bounded distance querying problem} is to efficiently compute the shortest path distance $d_{st}$ between $s$ and $t$ over $G'$ under the upper bound $d^{\top}_{st}$ such that,

\[
    d_{st}= 
\begin{cases}
    d_{G'}(s,t),& \text{if } d_{G'}(s,t)\leq d^{\top}_{st}\\
    d^{\top}_{st},              & \text{otherwise}
\end{cases}
\]

\end{definition}

In the following, we discuss the two steps of this framework in detail.

\subsection{Computing Upper Bounds}
Given any two vertices $s$ and $t$, we can use a \emph{highway cover distance labelling} $L$ to compute an upper bound $d^{\top}_{st}$ for the shortest path distance between $s$ and $t$ as follows,

\begin{align}\label{eq:upper-distance-bound}
d^{\top}_{st} = \texttt{min}\{\delta_L(r_i, s) + \delta_H(r_i, r_j) + \delta_L(r_j, t) |\notag\\(r_i, \delta_L(r_i, s)) \in  L(s),\notag\\ (r_j, \delta_L(r_j, t)) \in  L(t)\} 
\end{align}

This corresponds to the length of a shortest path from $s$ to $t$ passing through landmarks $r_i$ and $r_j$, where $\delta_L(r_i, s)$ is the shortest path distance from $r_i$ to $s$ in $L(s)$, $\delta_H(r_i, r_j)$ is the shortest path distance from $r_i$ to $r_j$ through highway $H$, and $\delta_L(r_j, t)$ is the shortest path distance from $r_j$ to $t$ in $L(t)$.

\begin{example}\label{exa:upperbound}
Consider the graph in Figure \ref{fig:highway-cover-labelling}(a), we may use the labels $L(2)$ and $L(11)$ to compute the upper bound for the shortest path distance between two vertices $2$ and $11$. There are two cases: (1) for the path $\langle 2, 5, 1, 11 \rangle$ that goes through landmarks 5 and 1, we have $\delta_L(5, 2) + \delta_H(5, 1) + \delta_L(1, 11) = 1 + 1 + 1 = 3$, and (2) for the path $\langle 2, 9, 1, 11 \rangle$ that goes through landmarks 9 and 1, we have $\delta_L(9, 2) + \delta_H(9, 1) + \delta_L(1, 11) = 2 + 1 + 1 = 4$. Thus, we take the minimum of these two distances as the upper bound, which is 3 in this case.
\end{example}

\subsection{Distance-Bounded Shortest Path Search}

We conduct a bidirectional search on the sparsified graph $G[V \backslash R]$ which is bounded by the upper bound $d^{\top}_{st}$ from the highway cover distance labelling. For a pair of vertices $\{s,t\} \subseteq V \backslash R$, we run breadth-first search algorithm from $s$ and $t$, simultaneously \cite{hayashi2016fully}. Algorithm \ref{algo:bidirectionalBFS} shows the pseudo-code of our distance-bounded shortest path search algorithm. We use two sets of vertices $\mathcal{P}_s$ and $\mathcal{P}_t$ to keep track of visited vertices from $s$ and $t$. We use two queues $\mathcal{Q}_s$ and $\mathcal{Q}_t$ to conduct both a forward search from $s$ and a reverse search from $t$. Furthermore, we use two integers $d_s$ and $d_t$ to maintain the current distances from $s$ and $t$, respectively.

During initialization, we set $\mathcal{P}_s$ and $\mathcal{P}_t$ to $\{s\}$ and $\{t\}$, and enqueue $s$ and $t$ into $\mathcal{Q}_s$ and $\mathcal{Q}_t$, respectively. In each iteration, we increment $d_s$ or $d_t$ and expand $\mathcal{P}_s$ or $\mathcal{P}_t$ by running either a forward search (FS) or a reverse search (RS) as long as $\mathcal{P}_s$ and $\mathcal{P}_t$ have no any common vertex or $d_s + d_t$ is equal to the upper bound $d^{\top}_{st}$, and $\mathcal{Q}_s$ and $\mathcal{Q}_t$ are not empty. In the forward search from $s$, we examine the neighbors $N_{G[V \backslash R]}(v)$ of each vertex $v \in \mathcal{Q}_s$. Suppose we are visiting a vertex $w \in N_{G[V \backslash R]}(v)$, if $w$ is included in vertex set $\mathcal{P}_t$, then it means that we find a shortest path to vertex $t$ of length $d_s + 1 + d_t$, because the reverse search from $t$ had already visited $w$ with distance $d_t$. At this stage, we return $d_s + 1 + d_t$ as the answer since we already know $d_s + d_t + 1 \leq d_G(s, t) \leq d^{\top}_{st}$. Otherwise, we add vertex $w$ to $\mathcal{P}_s$ and enqueue $w$ into a new queue $\mathcal{Q}_{new}$. When we could not find the shortest distance in the iteration, we replace $\mathcal{Q}_s$ with $\mathcal{Q}_{new}$ and increase $d_s$ by 1, and check if $d_s + d_t= d^{\top}_{st}$. If it holds, then we return $d^{\top}_{st}$ since $d^{\top}_{st} \leq d_G(s, t) \leq d_s + d_t + 1$.


\begin{algorithm}
 \SetAlgoLined
 \caption{Distance-Bounded Shortest Path Search}
 \label{algo:bidirectionalBFS}
 \KwIn{$G[V \backslash R]$, $s$, $t$, $d^{\top}_{st}$}
 \KwOut{$d_{G[V \backslash R]}(s, t)$}
 $\mathcal{P}_s \gets \{s\}$, $\mathcal{P}_t \gets \{t\}$, $d_s \gets 0$, $d_t \gets 0$ \\
 Enqueue $s$ to $\mathcal{Q}_s$, $t$ to $\mathcal{Q}_t$ \\
 \While{$\mathcal{Q}_s$ and $\mathcal{Q}_t$ are not empty}{
 	\uIf{$|\mathcal{P}_s| \leq |\mathcal{P}_t|$}{
      $found \gets \texttt{FS}(\mathcal{Q}_s)$
    }
    \Else{
      $found \gets \texttt{RS}(\mathcal{Q}_t)$
    }
    \uIf{found = true}{
      \textbf{return} $d_s + 1 + d_t$
    }
    \ElseIf{$d_s + d_t = d^{\top}_{st}$}{
    	\textbf{return} $d^{\top}_{st}$
    }
 }
 \textbf{return} $\infty$
\end{algorithm}

\begin{figure}[h]
\centering
\includegraphics[scale=0.3]{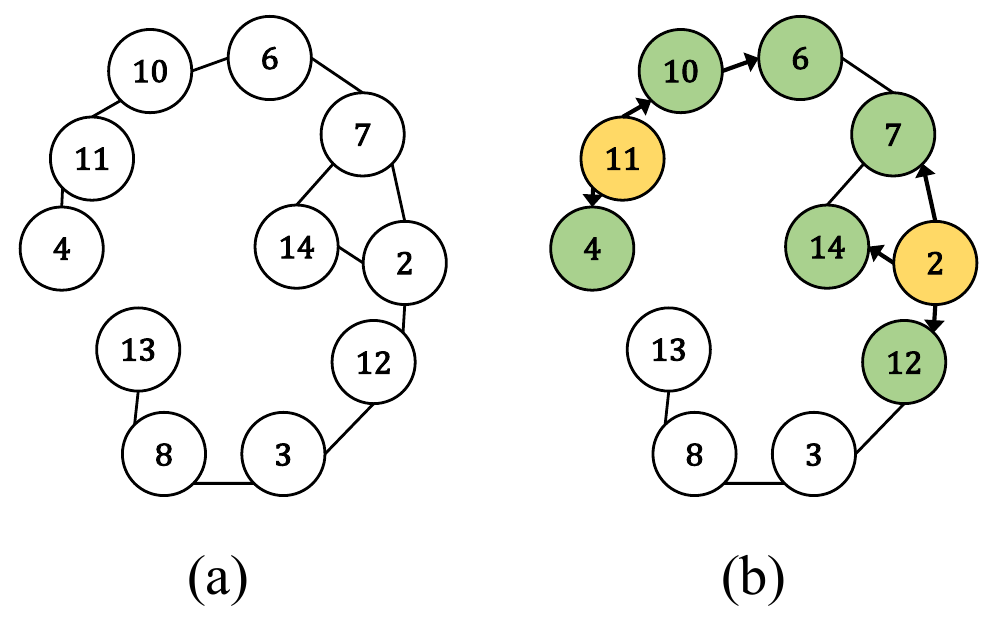}
\vspace{-0.2cm}
\caption{An illustration of the distance-bounded shortest path search algorithm \cite{hayashi2016fully}: (a) shows the sparsified graph after removing three landmarks $\{1, 5, 9\}$ from the graph in Figure \ref{fig:highway-cover-labelling}(a); (b) shows an example of computing the bounded distance between vertices $2$ and $11$ as highlighted in yellow, and green vertices denote the visited vertices in the forward and reverse searches.
}
\label{fig:sparsified-graph}
\end{figure}

\begin{example}
In Figure \ref{fig:sparsified-graph}(b), the upper distance bound between vertices 2 and 11 is 3, as computed in Example \ref{exa:upperbound}. Suppose that we run BFSs from vertices 2 and 11 respectively. First, a forward search from 2 enqueues its neighbors 7, 12 and 14 into $\mathcal{Q}_2$ and increases $d_2$ by 1. Then a reverse search from 11 enqueues 4 and 10 into $\mathcal{Q}_{11}$ and also sets $d_{11}$ to 1. At this stage, we have not found any common vertex between $\mathcal{Q}_2$ and $\mathcal{Q}_{11}$, and $d_2 + d_{11}=2$ which is less the upper bound $3$. Therefore, we continue to start a search from the vertices in $\mathcal{Q}_{11}$, which enqueues 5 into $\mathcal{Q}_{11}$ and increments $d_{11}$ to 2. Now, we have $d_2 + d_{11} == 3$ reaching the upper bound, hence we do not need to continue our search.

\end{example}

\subsection{Correctness}

The correctness of our querying framework can be proven based on the following two lemmas. More specifically, Lemma \ref{lemma:upper-distance-bound} can be derived by the highway cover property and the definition of $d^{\top}_{st}$. Lemma \ref{lemma:bounded-distance} can also be proven by the property of shortest path and the definition of the sparsified graph $G[V \backslash R]$.

\begin{lemma}\label{lemma:upper-distance-bound}
For a highway cover distance labelling $L$ over ($G,H$), we have $d^{\top}_{st}\geq d_G(s,t)$ for any two vertices $s$ and $t$ of $G$, where $d^{\top}_{st}$ is computed using $L$ and $H$.
\end{lemma}

\begin{lemma}\label{lemma:bounded-distance}
For any two vertices $\{s, t\}\subseteq V\backslash R$, if there is a shortest path between $s$ and $t$ in $G$ that does not include any vertex in $R$, then $d_G(s,t)=d_{G[V\backslash R]}(s,t)$ holds.
\end{lemma}

Thus, the following theorem holds:

\begin{theorem}
Let $G=(V,E)$ be a graph, $H$ a highway over $G$ and $L$ a highway cover distance labelling. Then, for any two vertices $\{s,t\}\subseteq V$, the querying framework over $(G, H, L)$ yields $d_G(s, t)$. 

\begin{proof}
We consider two cases: (1) $P_{st}$ contains at least one landmark. In this case, By Lemma \ref{lemma:upper-distance-bound} and the definition of the highway cover property, we have $d^{\top}_{st}=d_G(s,t)$. (2) $P_{st}$ does not contain any landmark. By Lemma \ref{lemma:bounded-distance}, we have $d_{G[V\backslash R]}(s,t)=d_G(s,t)$.
\end{proof}
\end{theorem}

\section{Optimization Techniques}\label{section:optimization}
In this section, we discuss optimization techniques for label construction, label compression, and query processing.

\begin{table*}[ht]
  \centering
  \caption{Datasets, where $|G|$ denotes the size of a graph $G$ with each edge appearing in the forward and reverse adjacency lists and being represented by 8 bytes. }
  \label{table:datasets}
  \scalebox{1}{
    \begin{tabular}{| l l l || r r r r r r ||c|} 
      \hline
      Dataset & Network & Type & $n$ & $m$ & $m/n$ & avg. deg & max. deg & $|G|$ & \hspace{0cm} Sources \hspace{0.5cm}\\
      \hline\hline
      Skitter & computer & undirected & 1.7M & 11M  & 6.5 & 13.081 & 35455 & 85 MB & \cite{konect:2017} \\
      Flickr & social & undirected & 1.7M & 16M & 9.1 & 18.133 & 27224 & 119 MB & \cite{konect:2017} \\
      Hollywood & social & undirected & 1.1M & 114M & 49.5 & 98.913 & 11467 & 430 MB & \cite{BoVWFI,BRSLLP} \\
      Orkut & social & undirected & 3.1M & 117M & 38.1 & 76.281 & 33313 & 894 MB & \cite{konect:2017} \\
      enwiki2013 & social & directed & 4.2M & 101M & 21.9 & 43.746 & 432260 & 701 MB & \cite{BoVWFI,BRSLLP} \\
      LiveJournal & social & directed & 4.8M & 69M & 8.8 & 17.679 & 20333 & 327 MB & \cite{konect:2017} \\ \hline
      Indochina & web & directed & 7.4M & 194M & 20.4 & 40.725 & 256425 & 1.1 GB & \cite{BoVWFI,BRSLLP} \\
      it2004 & web & directed & 41M & 1.2B & 24.9 & 49.768 & 1326744 & 7.7 GB & \cite{BoVWFI,BRSLLP} \\
      Twitter & social & directed & 42M & 1.5B & 28.9 & 57.741 & 2997487 & 9.0 GB & \cite{BoVWFI,BRSLLP} \\
      Friendster & social & undirected & 66M & 1.8B & 22.5 & 45.041 & 4006 & 13 GB & \cite{leskovec2015snap} \\ 
      uk2007 & web & directed & 106M & 3.7B & 31.4 & 62.772 & 979738 & 25 GB & \cite{BoVWFI,BRSLLP} \\
      ClueWeb09 & computer & directed & 2B & 8B & 5.98 & 11.959 & 599981958 & 55 GB & \cite{nr} \\ 
      \hline
    \end{tabular}
  }
\end{table*}

\subsection{Label Construction}
A technique called Bit-Parallelism (BP) has been previously used in several methods \cite{akiba2013fast, hayashi2016fully} to speed up the label construction process. The key idea of BP is to perform BFSs from a given landmark $r$ and up to 64 of its neighbors simultaneously, and encode the relative distances (-1, 0 or 1) of these neighbors w.r.t. the shortest paths between $r$ and each vertex $v$ into a 64-bit unsigned integer. In the work \cite{akiba2013fast}, 
BP was applied to construct bit-parallel labels from initial vertices without pruning, which aimed to leverage the information from these bit-parallel labels to cover more shortest paths between vertices. Then, both bit-parallel labels and normal labels are constructed in the pruned BFSs.
The work in \cite{hayashi2016fully} also used BP to construct thousands of bit-parallel shortest-path trees (SPTs) because it is very costly to construct thousands of normal SPTs in memory owing to their prohibitively large space requirements and very long construction time. 

In our work, we develop a simple yet rigorous parallel algorithm (HL-P) which can run parallel BFSs from multiple landmarks (depending on the number of processors) to construct labelling in an extremely efficient way for massive networks, with much less time as will be demonstrated in our experiments.  

\subsection{Label Compression}
The choice of the data structure for labels may significantly affect the performance of index size and memory usage. As noted in \cite{li2017experimental}, some works \cite{abraham2012hierarchical,delling2014robust} did not elaborate on what data structure they have used for representing labels. Nonetheless, for the works that are most relevant to ours, such as FD \cite{hayashi2016fully} and PLL \cite{akiba2013fast}, they used 32-bit integers to represent vertices and 8-bit integers to represent distances for normal labels. In addition to this, they also used 64-bits to encode the distances from a landmark to up to 64 of its neighbors in their shortest paths to other vertices. Since our approach only selects a very small number of landmarks to construct the highway cover labelling (usually no more than 100 landmarks), we may use 8 bits to represent landmarks and another 8 bits to store distances for labels. In order to fairly compare methods from different aspects, we have implemented our methods using both 32 bits and 8 bits for representing vertices in labels. However, different from the BP technique that uses 64-bits to encode the distance information of up to 64 neighbours of a landmark, our parallel algorithm (HL-P) does not use a different data structure for labels constructed in parallel BFSs.  

\subsection{Query Processing}
We show that computing the upper bound $d^{\top}_{st}$ can be optimized based on the observation, captured by the following lemma.

\begin{lemma}
For a highway cover distance labelling $L$ over $(G,H)$, where $G=(V,E)$ and $H=(R, \delta_H)$, and any $\{s,t\} \subseteq V\backslash R$, if a landmark $r$ appears in both $L(s)$ and $L(t)$, then $\delta_L(r, s) + \delta_L(r, t)\leq \delta_L(r, s) + \delta_H(r, r') + \delta_L(r', t)$ holds for any other $r'\in R$. 
\end{lemma}
\begin{proof}
By the definition of the highway cover property, we know that $r$ is not in the shortest path between $r'$ and $t$. Then by triangle inequality in Equation \ref{distance_metric}, this lemma can be proven.
\end{proof}

Thus, in order to efficiently compute the upper bound $d^{\top}_{st}$, for any landmarks that appear in both $L(s)$ and $L(t)$, we compute the $r$-constrained shortest path distance between $s$ and $t$ using Equation \ref{eq:twoHop_query_distance}, while for a landmark $r'$ that only appear in one of $L(s)$ and $L(t)$, we use Equation \ref{eq:upper-distance-bound} to calculate the $r'$-constrained shortest path distance between $s$ and $t$. This would lead to more efficient computations for queries when the landmarks appear in both labels of two vertices.

\section{EXPERIMENTS}\label{section:experiments}
To compare the proposed method with baseline approaches, we have implemented our method in C++11 using STL libraries and compiled using gcc 5.5.0 with the -O3 option. We performed all the experiments using a single thread on a Linux server (having 64 AMD Opteron(tm) Processors 6376 with 2.30GHz and 512GB of main memory) for sequential version of the proposed method and up to 64 threads for parallel version of the proposed method.

\begin{figure}[h]
\centering
\includegraphics[scale=0.45]{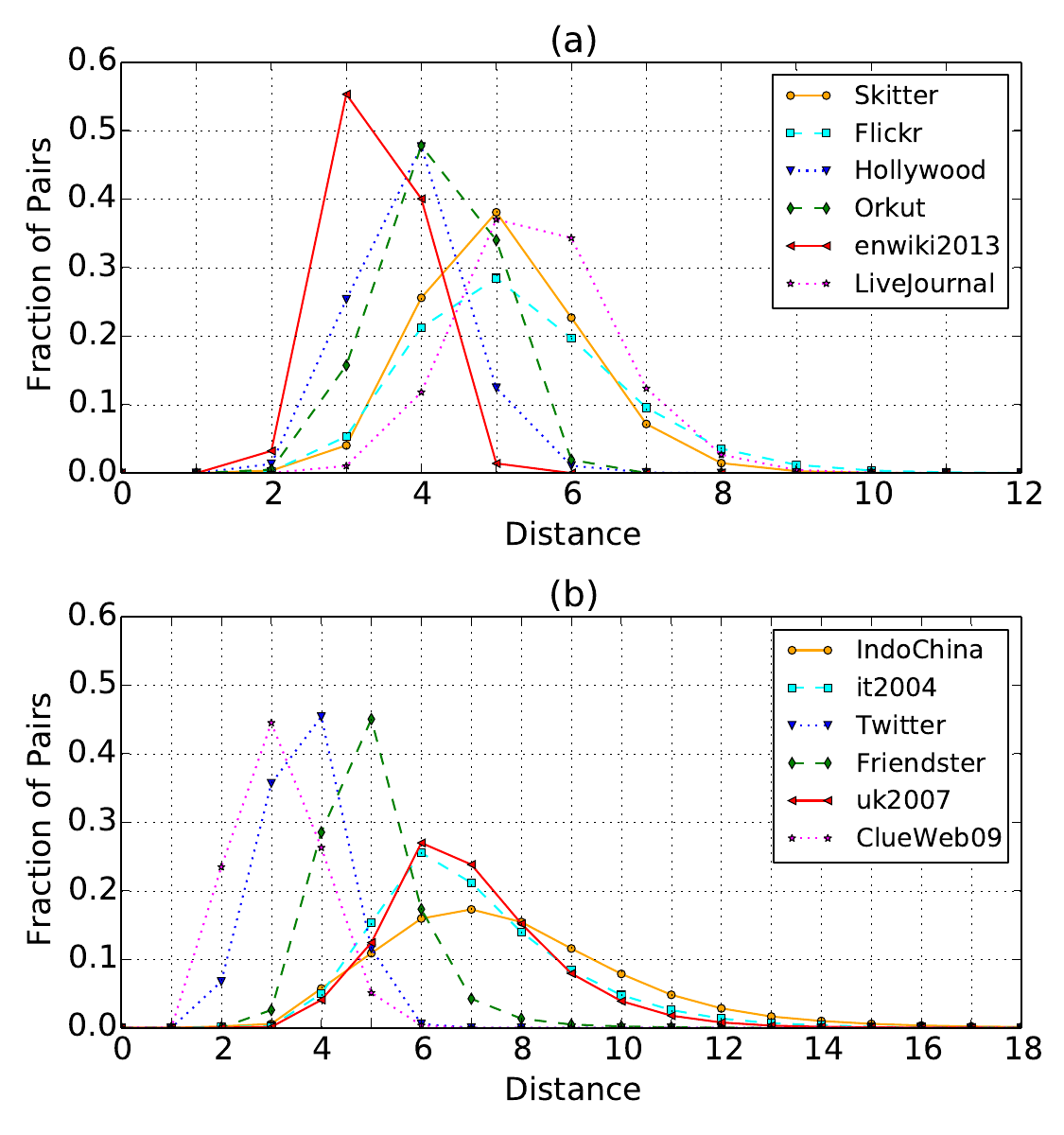}
\vspace{-0.7cm}
\caption{Distance distribution of 100,000 random pairs of vertices on all the datasets.}\vspace{-0cm}
\label{fig:distance-distribution}
\end{figure}

\begin{table*}[h!]
 \centering
 \caption{Comparison of construction times and query times between our methods, i.e., HL-P and HP, and the state-of-the-art methods, where CT denotes the CPU clock time in seconds for labelling construction, QT denotes the average query time in milliseconds, and ALS denotes the average number of entries per label.}
 \label{table:performance}
 \begin{tabular}{| l || r r r r r | r r r r | r | r r r r |}  \hline
	\multirow{2}{*}{Dataset} & \multicolumn{5}{c|}{CT[s]}&\multicolumn{5}{c|}{QT[ms]}&\multicolumn{4}{c|}{ALS} \\\cline{2-15}
    & HL-P & HL & FD & PLL & IS-L & HL & FD & PLL & IS-L & Bi-BFS & HL & FD & PLL & IS-L \\

   \hline\hline
	Skitter & 2&13  & 30 & 638 & 1042 & 0.067 & 0.043 & 0.008 & 3.556 & 3.504 & 12 & 20+64 & 138+50 & 51 \\
    Flickr & 2&14  & 41 & 1330 & 8359 & 0.015 & 0.028 & 0.01 & 33.760  & 4.155 & 10 & 20+64& 290+50 & 50 \\
    Hollywood & 3& 17  & 107 & 31855 & DNF & 0.047 & 0.075 & 0.051 & - & 6.956 & 12 & 20+64 &2206+50 & - \\
    Orkut & 10&62  & 366 & DNF & DNF & 0.224 & 0.251 & - & - & 21.086 & 11 &20+64 & - & - \\
    enwiki2013 & 9& 77  & 308 & 22080 & DNF & 0.190 & 0.131 & 0.027 & - & 19.423 & 10 &20+64& 471+50 & - \\
    LiveJournal & 9&77  & 166 & DNF & 20583 & 0.088 & 0.111 & - & 56.847 & 17.264 & 13 &20+64& - & 69 \\ \hline
    Indochina & 8&50  & 144 & 9456 & DNF & 1.905 & 1.803 & 0.02 & - & 9.734 & 5 &20+64& 441+50 & - \\
    it2004 & 66&304  & 1623 & DNF & DNF & 2.684 & 2.118 & - & - & 92.187 & 10 &20+64& - & - \\
    Twitter & 133&1380  & 1838 & DNF & DNF & 1.424 & 0.432 & - & - & 426.949 & 14 &20+64& - & - \\
    Friendster & 135&2229  & 9661 & DNF & DNF & 1.091 & 1.435 & - & - & 534.576 & 19 &20+64& - & - \\
    uk2007 & 110&1124  & 6201 & DNF & DNF & 11.841 & 18.979 & - & - & 355.688 & 8 &20+64& - & - \\
    ClueWeb09 & 4236& 28124  & DNF & DNF & DNF & 0.309 & - & - & - & - & 2 &- &- & - \\
   \hline
 \end{tabular}
 \end{table*}
 \begin{table}
  \caption{Comparison of labelling sizes between our methods, i.e., HL(8) and HL, and the state-of-the-art methods.}
 \label{table:index-size}
 \begin{tabular}{| l || r r r r r |}  \hline
   Dataset & HL(8) & HL & FD & PLL & IS-L \\
   \hline\hline
   Skitter & 42MB & 102MB  & 202MB & 2.5GB & 507MB \\
   Flickr & 34MB & 81MB  & 178MB & 3.7GB & 679MB \\
   Hollywood & 28MB &67MB & 293MB & 13GB & - \\
   Orkut & 70MB & 170MB & 756MB & - & - \\
   enwiki2013 & 83MB & 200MB  & 743MB & 12GB & - \\
   LiveJournal & 123MB & 299MB  & 778MB & - & 3.8GB \\ \hline
   Indochina & 81MB & 192MB & 999MB & 21GB & - \\
   it2004 & 855MB & 2GB & 5.6GB & - & - \\
   Twitter & 1.2GB & 2.8GB  & 4.8GB & - & - \\
   Friendster & 2.5GB & 5.2GB & 11.8GB & - & - \\
   uk2007 & 1.8GB & 4.3GB & 14.1GB & - & - \\
   ClueWeb09 & 4.7GB & 9GB  & - & - & - \\
   \hline
 \end{tabular}\vspace{-0.1cm}
\end{table}

\subsection{Datasets}
In our experiments, we used 12 large-scale real-world complex networks, which are detailed in Table \ref{table:datasets}. These networks have vertices and edges ranging from millions to billions. Among them, the largest network is ClueWeb09 which has 2 billions of vertices and 8 billions of edges. We included this network in our experiments for the purpose of evaluating the robustness and scalability of the proposed method. In previous works, the largest dataset that has been reported is uk2007 which has only around 100 millions of vertices and 3.7 billions of edges. For all these networks, we treated them as undirected and unweighted graphs. 

To investigate the query time of finding the distance between two vertices, we randomly sampled 100,000 pairs of vertices from all pairs of vertices in each network, i.e., $V\times V$. The distance distribution of these 100,000 randomly sampled pairs of vertices are shown in Figure \ref{fig:distance-distribution}(a)-\ref{fig:distance-distribution}(b), from which we can confirm that most of pairs of vertices in these networks have a small distance ranging from 2 to 8.


\subsection{Baseline Methods}
We compared our proposed method with three state-of-the-art methods. Two of these methods, namely fully dynamic (FD) \cite{hayashi2016fully} and IS-L \cite{fu2013label}, combine a distance labelling algorithm with a graph traversal algorithm for distance queries on complex networks. The third one is pruned landmark labelling (PLL) \cite{akiba2013fast} which is completely based on distance labelling to answer distance queries. Besides these, there are a number of other methods for answering distance queries, such as 
HDB \cite{jiang2014hop}, RXL and CRXL \cite{delling2014robust}, HCL \cite{jin2012highway}, HHL \cite{abraham2012hierarchical} and TEDI \cite{wei2010tedi}. However, since the experimental results of the previous works \cite{hayashi2016fully,akiba2013fast} have shown that FD outperforms HDB, RXL and CRXL, and PLL outperforms HCL, HHL and TEDI, we omit the comparison with these methods. 

In our experiments, the implementations of the baseline methods FD, IS-L and PLL were provided by their authors, which were all implemented in C++. We used the same parametric settings for running these methods as suggested by their authors. For instance, the number of landmarks is chosen to 20 for FD \cite{hayashi2016fully}, the number of bit-parallel BFSs is set to 50 for PLL \cite{akiba2013fast}, and $k$ is 6 for graphs larger than 1 million vertices for IS-L \cite{fu2013label}.

\subsection{Comparison with Baseline Methods}
To evaluate the performance of our proposed approach, we compared our approach with the baseline methods in terms of the construction time of labelling, the size of labelling, and querying time. The experimental results are presented in Tables \ref{table:performance} and \ref{table:index-size}, where DNF denotes that a method did not finish in one day or ran out of memory. In order to make a consistent comparison with the baseline methods \cite{hayashi2016fully,akiba2013fast,fu2013label}, we chose top 20 vertices as landmarks after sorting based on decreasing order of their degrees, and also used 32-bit integers to represent vertices and 8-bit integers to represent distances.

\subsubsection{Construction Time}
As shown in Table \ref{table:performance}, our proposed method (HL) has successfully constructed the distance labelling on all the datasets for a significantly less amount of time than the state-of-the-art methods. As compared to FD, our method is on average 5 times faster and have results on all the datasets. In contrast to this, FD failed to construct labelling for the largest dataset ClueWeb09. PLL failed for 7 out of 12 datasets, including the datasets Orkut and LiveJournal which have less than 120 millions of edges, due to its prohibitively high preprocessing time and memory requirements for building labelling. IS-L failed to construct labelling for all the datasets that have edges more than 100 million due to its very high cost for computing independent sets on massive networks, i.e. failed for 9 out of 12 datasets. We can also see from Table \ref{table:performance} that the parallel version of our method (HL-P) is much faster than the sequential version (HL). Compared with FD, HL-P is more than 50-70 times faster for the two large datasets Friendster and uk2007. This confirms that our method can construct labelling very efficiently and is scalable on large networks with billions of vertices and edges.

\begin{figure*}[ht]
\centering
\includegraphics[width=\textwidth]{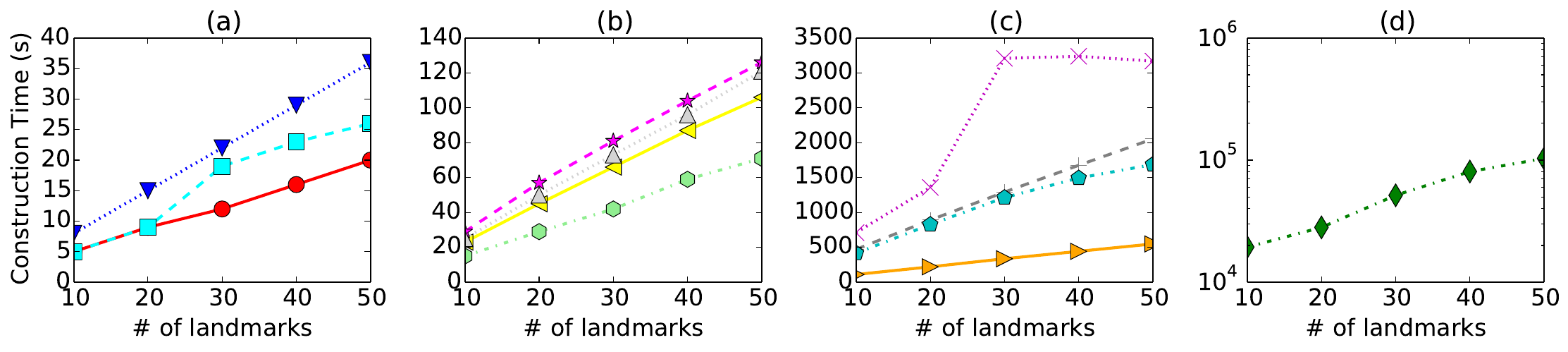}
\includegraphics[width=\textwidth]{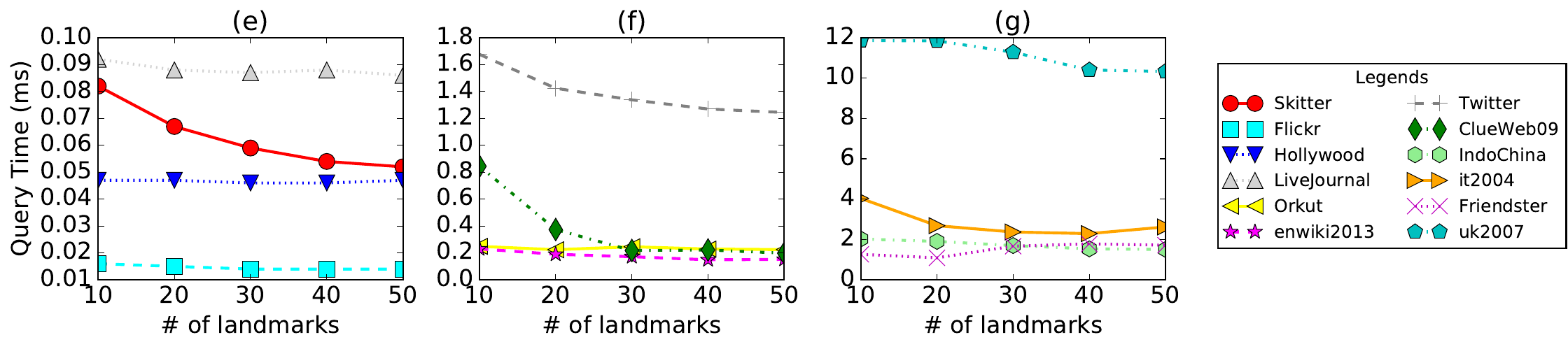}
\vspace{-0.5cm}
\caption{(a)-(d) Construction times using our method HL under 10-50 landmarks on all the datasets; (e)-(g) Query times using our method HL under 10-50 landmarks on all the datasets.}
\label{fig:querying-times}
\end{figure*}

\subsubsection{Labelling Size}
As we can see from Table \ref{table:index-size} that the labelling sizes of all the datasets constructed by the proposed method are significantly smaller than the labelling sizes of FD and much smaller than PLL and IS-Label. Specifically, our labelling sizes using 32-bits representation of vertices (HL) are 2-5 times smaller than FD except for ClueWeb09 (as discussed before, FD failed to construct labelling for ClueWeb09), 7 times smaller than IS-Label on Skitter, Flickr and LiveJournal and more than 60 times smaller than PLL for Skitter, Flickr, Hollywood, enwiki2013 and Indochina. The compressed version of our method that uses 8-bits representation of vertices (i.e. HL(8)) produces further smaller index sizes as compared to uncompressed version (HL). 
Here, It is important to note that the labelling sizes of almost all the datasets are also significantly smaller than the original sizes of the datasets shown in Table \ref{table:datasets}. This also shows that our method is highly scalable on large networks in terms of the labellng sizes.

\subsubsection{Query Time}\label{subsubsection:queryTime}
the  This is due to a very small average labelling size (i.e., 12) as compared with FD and PLL (i.e., 20+64 and 2206+50, respectively) and a very small average distance. The average query time of HL on Twitter is 3 times slower than FD. This may be due to a large portion of covered pairs by FD as shown in Figure \ref{fig:covered-pairs} which contributes towards an effective bounded traversal on the sparsified network since the landmarks of Twitter have very high degrees and the average distance is also very small. Moreover, the average query times of HL and FD on Indochina, it2004, Friendster and uk2007 are more than 1ms due to comparatively large average distances than other datasets as shown in Figure \ref{fig:distance-distribution}(b). Note that all the baseline methods are not scalable enough to have results for ClueWeb09 and the average query time on ClueWeb09 of our method HL is small because of a very large portion of covered pairs and a small average label size. We also reported the average query time for online bidirectional BFS algorithm (Bi-BFS) using randomly selected 1000 pairs of vertices in Table \ref{table:performance}. As we can see that Bi-BFS has considerably long query times, which are not practicable in applications for performing distance queries in real time.

\subsection{Performance under Varying Landmarks}
We have also evaluated the performance of our method (HL) by varying the number of landmarks between 10 and 50, which are again selected based on highest degrees.

\subsubsection{Construction Time}
The construction times of our method HL against different numbers of landmarks (from 10 to 50) are shown in Figure \ref{fig:querying-times}(a)-\ref{fig:querying-times}(d). We can see that the construction times are linear in terms of the number of landmarks, which confirms the scalability of our method. In Figure \ref{fig:querying-times}(a)-\ref{fig:querying-times}(b), our method is able to construct labelling for 7 datasets under 50 landmarks from 20 seconds to 2 minutes, which is not possible with any state-of-the-art methods. In Figure \ref{fig:querying-times}(c), the construction time using 50 landmarks of Friendster is 3 times faster and the construction time of uk2007 is 4 times faster than FD using only 20 landmarks as shown in Table \ref{table:performance}. Figure \ref{fig:querying-times}(d) shows the construction time for ClueWeb09 which has 2 billion vertices and 8 billion edges. The significant improvement in construction time allows us to compute labelling for a large number of landmarks, leading to better pair coverage ratios to tighten upper distance bounds (will be further discussed in Section \ref{pair-coverage}).

\begin{figure*}[ht]
\centering
\includegraphics[width=\textwidth]{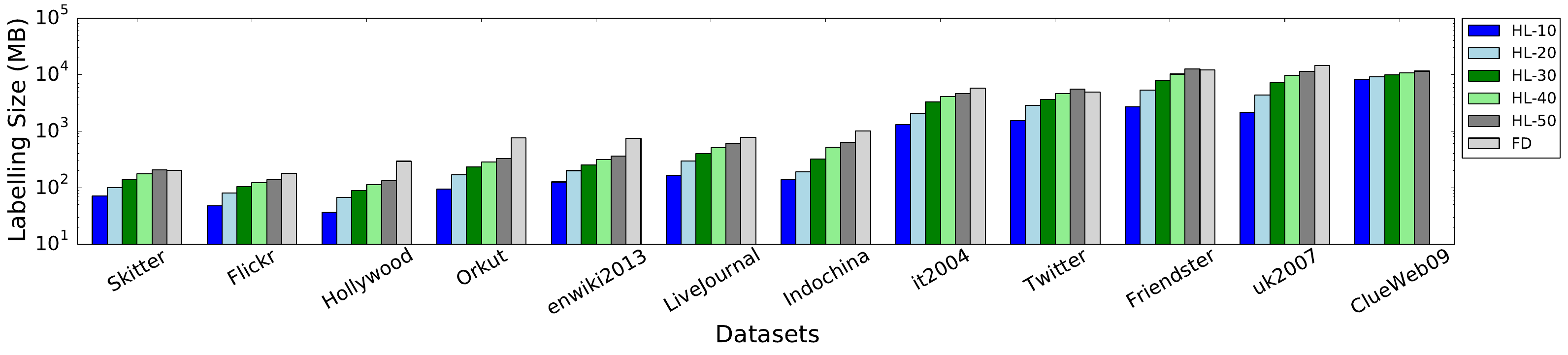}
\vspace{-0.7cm}
\caption{Labelling sizes using our method HL under 10-50 landmarks and using FD on all the datasets.}
\label{fig:labelling-size}
\vspace{0.4cm}
\includegraphics[width=\textwidth]{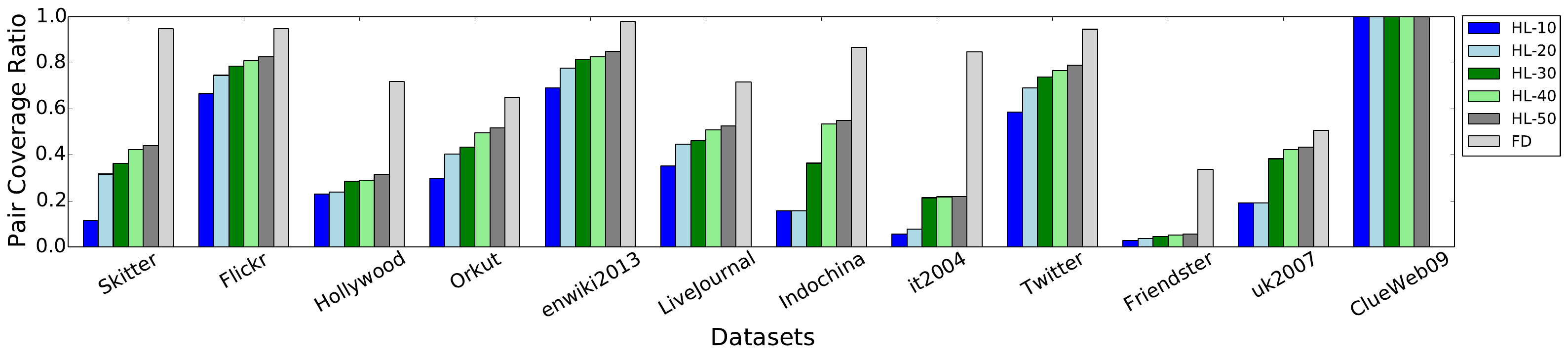}
\vspace{-0.7cm}
\caption{Pair coverage ratios using our method HL under 10-50 landmarks and using FD on all the dataset.}
\label{fig:covered-pairs}
\end{figure*}

\subsubsection{Labelling Size}
Figure \ref{fig:labelling-size} shows the labelling sizes of HL using 10, 20, 30, 40 and 50 landmarks on all the dataset, and of FD using only 20 landmarks on all the datasets except for ClueWeb09 (as discussed before, FD failed to construct labeling for ClueWeb09). It can be seen that the labelling sizes of HL increase linearly with the increased number of landmarks, and even the labelling sizes of HL using 50 landmarks are almost always smaller than the labelling sizes constructed by FD using only 20 landmarks. This reduction in labelling sizes enables us to save space and memory, thus makes our method scalable on large networks.

\subsubsection{Query Time}
Figure \ref{fig:querying-times} shows the impact of using different numbers of landmarks between 10 and 50 on average query time of our method. The average query times either decrease or remain the same when the number of landmarks increases, except for Orkut when using 30 landmarks and for Friendster when using landmarks greater than 20. In particular, on Friendster, labelling sizes are very large as shown in Figure \ref{fig:labelling-size} and the fraction of covered pairs (i.e., pair coverage ratio) is very small as shown in Figure \ref{fig:covered-pairs}, which may have slowed down our query processing due to a longer time for computing upper distance bounds and ineffective use of bounded-distance traversal.

\subsubsection{Pair Coverage}\label{pair-coverage}
Figure \ref{fig:covered-pairs} presents the ratios of pairs of vertices covered by at least one landmark (i.e., pair coverage ratios) in HL using 10-50 landmarks and in FD using 20 landmarks. As we can observe that the pair coverage ratios for HL increase when the number of landmarks increases and 40 turns out to be the better choice on the number of landmarks for most of the datasets. Specifically, pair coverage ratios on Orkut, enwiki2013, Indochina and uk2007 with 40 landmarks are good, resulting in better query times than using 20 landmarks, as shown in Figure \ref{fig:querying-times}. On datasets such as Hollywood and it2004, 30 landmarks are a better option than 40 landmarks because they only slightly differ in the pair coverage ratios and query times w.r.t. using 40 landmarks, but with reduced labelling sizes. The pair coverage ratios by FD are greater than HL on all the datasets except for ClueWeb09, which may be the reason behind its better query times for some datasets as shown in Table \ref{table:performance}. Note that, on ClueWeb09, we obtain almost hundred percentage for pair coverage due to its very high degree landmarks.

\section{RELATED WORK}\label{section:background}


A naive solution for exact shortest-path distance computation is to run the Dijkstra search for weighted graphs or BFS for unweighted graphs, from a source vertex to a destination vertex \cite{tarjan1983data}. To improve search efficiency, a bidirectional scheme can be used to run two such searches: one from the source vertex and the other from the destination vertex \cite{pohl1971bi}. Later on, Goldberg et al. \cite{goldberg2005computing} combined the bidirectional search technique with the A* algorithm to further improve the search performance. In their method, they precomputed labeling based on landmarks to estimate the lower bounds, and used that estimate with a bidirectional A* search for efficient computation of shortest-path distances. However, this method is known to work only for road networks and do not scale well on complex networks \cite{hayashi2016fully}.

To efficiently answer exact shortest-path distance queries on graphs, labelling-based methods have been developed with great success \cite{akiba2013fast,abraham2012hierarchical,fu2013label,jin2012highway,abraham2011hub,li2017experimental}. Most of them construct a labeling based on the idea of 2-hop cover \cite{cohen2003reachability}. It has also been shown that computing a minimal 2-hop cover labeling is NP-hard \cite{abraham2012hierarchical,cohen2003reachability}. In \cite{abraham2011hub}, the authors proposed a hub-based labeling algorithm (HL) which constructs hub labelling by processing contraction hierarchies (CH) and is among the fastest known algorithms for distance queries in road networks. However, the method is not feasible for complex networks as reported by the same authors and they thus proposed a hierarchical hub-labeling (HHL) algorithm for complex networks in \cite{abraham2012hierarchical}. In this work, a top-down method was used to maintain a shortest-path tree for every vertex in order to indicate all uncovered shortest-paths at each vertex. Due to very high storage and computation requirements, the method is also not scalable for handling large graphs. Another method called Highway Centric Labeling (HCL) was proposed by Jin et al. \cite{jin2012highway} which exploits highway structure of a graph. This method aimed to find a spanning tree which can assist in optimal distance labelling and used that spanning tree as a highway to compute a highway-based 2-hop labelling for fast distance computation. After that, in \cite{akiba2013fast}, Akiba et al. proposed the pruned landmark labeling (PLL) method which precomputes a distance-aware 2-hop cover index by performing a pruned breadth-first search (BFS) from every vertex. The idea is to prune vertices whose distance information can be obtained using a partially available 2-hop index constructed via previous BFSs. This work helps to achieve low construction cost and smaller index size due to reduced search space on million-scale networks. It has been shown that PLL outperforms other state-of-the-art methods available at the time of publication, including HHL \cite{abraham2012hierarchical}, HCL \cite{jin2012highway}  and TEDI \cite{wei2010tedi}. However, PLL is still not feasible for constructing 2-hop cover indices for billion-scale networks due to a very high memory requirement for labelling construction.

Fu et al. \cite{fu2013label} proposed IS-Label (IS-L) which gained significant scalability in precomputing 2-hop cover distance labellings for large graphs with hundreds of millions of vertices and edges. IS-L uses the notion of an independent set of vertices in a graph. First, it computes an independent set of vertices from a graph, then it constructs a graph by removing the independent set of vertices from the previous graph recursively and augments edges that preserve distance information after the removal of the independent set of vertices. All the vertices in the remaining graph preserve their distance information to/from each other. Generally, IS-L is regarded as a hybrid method that combines distance labelling with graph traversal for complex networks \cite{li2017experimental}. Following the same line of thought, 
very recently, Akiba et al. \cite{hayashi2016fully} proposed a method to accelerate shortest-path distances computation on large-scale complex networks. To the best of our knowledge, this work is most closely related to our work presented in this paper. The key idea of the method in \cite{hayashi2016fully} is to select a small set of landmarks $R$ and precompute shortest-path trees (SPTs) rooted at each $r \in R$. Given any two vertices $s$ and $t$, it first computes the upper bound by taking the minimum length among the paths that pass through $R$. Then a bidirectional BFS from $s$ to $t$ is conducted on the subgraph $G \backslash R$ to compute the shortest-path distances that do not pass through $R$ and take the minimum of these two results as the answer to an exact distance query. The experiments in \cite{hayashi2016fully} showed that this method can scale to graphs with millions of vertices and billions of edges, and outperforms the state-of-the-art exact methods PLL \cite{akiba2013fast}, HDB \cite{jiang2014hop}, RXL and CRXL \cite{delling2014robust} with significantly reduced construction time and index size, while the query times are higher but still remain among 0.01-0.06 for most of graphs with less than 5M vertices.

Although the method proposed in \cite{hayashi2016fully} has been tested on a large network with millions of vertices and billions of edges, it still fails to construct labelling on billion-scale networks in general, particularly with billions of vertices. In contrast, our proposed method not only constructs labellings linearly with the number of landmarks in large networks with billions of vertices, but also enables the sizes of labellings to be significantly smaller than the original network sizes. In addition to these, the deterministic nature of labelling allows us to achieve further gains in computational efficiency using parallel BFSs over multiple landmarks, which is highly scalable for handling billion-scale networks. \looseness=-1
\section{CONCLUSION}\label{section:conclusion}
\vspace{-0.07cm}
We have presented a scalable solution for answering exact shortest path distance queries on very large (billion-scale) complex networks. The proposed method is based on a novel labelling algorithm that can scale to graphs at the billion-scale, and a querying framework that combines a highway cover distance labelling with distance-bounded searches to enable fast distance computation. We have proven that the proposed labelling algorithm can construct HWC-minimal labellings that are independent of the ordering of landmarks, and have further developed a parallel labelling method to speed up the labelling construction process by conducting BFSs simultaneously for multiple landmarks. The experimental results showed that the proposed methods significantly outperform the state-of-the-art methods. For future work, we plan to investigate landmark selection strategies for further improving the performance of labelling methods.  


\bibliographystyle{ACM-Reference-Format}
\bibliography{references}

\end{document}